\documentclass[preprint,11pt,authoryear]{elsarticle}
\usepackage{tikz}
\usepackage{xcolor}
\usepackage{amsfonts}
\usepackage{amsmath}
\usepackage{amsthm}
\usepackage{amssymb}
\usepackage{times}
\usepackage{booktabs}
\usepackage{multirow}
\usepackage{booktabs}
\usepackage{setspace}
\usepackage{enumerate}
\usepackage{natbib}
\usepackage{color}
\usepackage{colortbl}
\usepackage{ wasysym }
\usepackage{verbatim}
\usetikzlibrary{shapes,arrows,calc,positioning}
\usepackage{lscape}
\usepackage{array}
\usepackage{setspace}
\usepackage{kotex}
\usepackage{url}
\usepackage{arydshln}

\usepackage{caption, subcaption}
\usepackage{rotating}
\usepackage[notcite, notref]{showkeys}
\newcolumntype{x}[1]{%
	>{\centering\hspace{0pt}}p{#1}}%

\definecolor{Gray}{gray}{0.9}
\newcolumntype{g}{>{\columncolor{Gray}}c}

\def\i.i.d.{\buildrel {\rm i.i.d.} \over \sim}

\def\cw#1 { \overset{\mathbb{P}}{\underset{#1}{\longrightarrow}} }
\def\Real{\mathbb{R}}
\def\Natu0{\mathbb{Z}_{\ge 0}}
\def\P#1{{\mathbb{P}}\left(#1\right)}

\def\E#1{{\mathbb E}\left[#1\right]}

\def \rcov#1#2 {{\rm cov}_{#1}\left( #2\right)}

\newtheorem{example}{Example}

\newtheorem{lemma}{Lemma}
\newtheorem{theorem}{Theorem}

\newtheorem{corollary}{Corollary}
\newtheorem{remark}{Remark}
\newtheorem{proposition}{Proposition}

\newtheorem*{toy*}{Toy Model}

\newtheorem{model}{Model}

\def\cov#1{{\rm  cov}\left[#1\right]}

\definecolor{Gray}{gray}{0.9}

\begin{document}
	
	\begin{frontmatter}
   \title{Counting models with excessive zeros ensuring stochastic monotonicity}
		
		%\author[IH]{Woojoo Lee}
		%\ead{lwj221@gmail.com}
		%\address[IH]{Department of Statistics, Inha University, 235 Yonghyun-Dong, Nam-Gu, Incheon 402-751, Korea.}
		%\author[BSNU]{Sojung C. Park}
		%\ead{sojungpark@snu.ac.kr}
		%\address[BSNU]{College of Business Administration, Seoul National University, Gwanak-gu, Seoul 151-916, Republic of Korea.}

		\author[EH]{Hyemin Lee \corref{cor1}}
		\ead{hyemin@bok.or.kr}
		\address[EH]{Department of Statistics, Ewha Womans University, Seoul, Republic of Korea}
		%\address[SF]{Department of Statistics and Actuarial Science, Simon Fraser University, BC, Canada}
		
		\author[EH]{Dohee Kim \corref{cor1}}
		\ead{doheekim22@ewha.ac.kr}
		%\address[EHM]{Department of Mathematics, Korea Military Academy, Seoul, Republic of Korea}
		
		\author[TW]{Banghee So \corref{cor2}}
		\ead{bso@towson.edu}
        \address[TW]{Department of Mathematics, Towson University, 7800 York Rd, Towson, MD, 21252, USA}

		\author[EH]{Jae Youn Ahn\corref{cor2}}
		\ead{jaeyahn@ewha.ac.kr}
		
		%\address[WI]{Wisconsin School of Business, University of Wisconsin-Madison, Madison, WI 53706, USA.}
		\cortext[cor1]{First Authors}
		\cortext[cor2]{Corresponding Author}

		%\author[EH]{Kyongwon Kim\corref{cor2}}
		%\ead{kimk@ewha.ac.kr}
		%\cortext[cor2]{Corresponding Authors}

\begin{abstract}
Standard count models such as the Poisson and Negative Binomial models often fail to capture the large proportion of zero claims commonly observed in insurance data. To address such issue of excessive zeros, zero-inflated and hurdle models introduce additional parameters that explicitly account for excess zeros, thereby improving the joint representation of zero and positive claim outcomes. These models have further been extended with random effects to accommodate longitudinal dependence and unobserved heterogeneity. However, their consistency with fundamental probabilistic principles in insurance, particularly stochastic monotonicity, has not been formally examined. This paper provides a rigorous analysis showing that standard counting random-effect models for excessive zeros may violate this property, leading to inconsistencies in posterior credibility. We then propose new classes of counting random-effect models that both accommodate excessive zeros and ensure stochastic monotonicity, thereby providing fair and theoretically coherent credibility adjustments as claim histories evolve.
\end{abstract}


	\end{frontmatter}
	
	\vfill
	
	\pagebreak

	\vfill
	
	\pagebreak
	
	%-----------------------------------------------------------------------------------------------------------------------------------------------------------
	%-----------------------------------------------------------------------------------------------------------------------------------------------------------
	\section{Introduction}
In insurance analytics, modeling claim frequency is essential for accurately determining premiums, allocating reserves, and assessing underwriting risks. This task is particularly challenging in areas such as automobile, homeowners, and health insurance, where many policyholders do not file any claims during the coverage period. The resulting large number of zero claims has been widely documented in actuarial and statistical studies. Although standard count models based on the Poisson or Negative Binomial distributions can accommodate zeros, they lack a separate component that specifically explains the occurrence of zeros. Consequently, these models often fit the data poorly when zeros are overly frequent, as adjusting the parameters to capture the zero part typically worsens the fit for the positive counts. As emphasized by \citet{hilbe2011negative} and \citet{perumean2013zero}, explicitly modeling the excess of zeros is crucial for obtaining reliable inference for both zero and nonzero outcomes. To address this issue, the statistical literature has developed two main classes of models: \textbf{zero-inflated models} and \textbf{hurdle models}.

Zero-inflated models use a two-component structure that combines a binary process with a standard count distribution \citep{lambert1992zero}, thereby distinguishing between structural zeros generated by the binary process and sampling zeros arising from the count distribution itself. The most common examples are the zero-inflated Poisson and zero-inflated Negative Binomial models. Building on this idea, many studies have extended and applied zero-inflated approaches in insurance \citep{yip2005modeling, boucher2009number, mouatassim2012poisson, chen2019subgroup, so2024enhanced, so2025advancing, so2025zeroinflatedTweedie}.

The hurdle model, also referred to as a two-part model \citep{heilbron1994zero}, adopts a two-stage structure consisting of a hurdle stage and a post-hurdle stage \citep{cragg1971some, mullahy1986specification, winkelmann2008econometric}.
The most common specification places the hurdle at zero: in the first stage, a binary model determines whether an observation crosses the hurdle, and conditional on this, the second stage models the positive counts using a standard count distribution.
Similar to zero-inflated models, the most widely used hurdle specifications in insurance are the Poisson hurdle and Negative Binomial hurdle models, which have been applied in various actuarial studies \citep{boucher2008modelling, zhang2022new}.

Despite their different formulations, both zero-inflated and hurdle models aim to address count data with an excessive number of zeros.
For clarity, we collectively refer to them as \textbf{counting models for excessive zeros}.
A key feature of these models is the inclusion of additional parameters that explicitly govern the generation of zeros, thereby allowing for a better joint representation of zero and positive outcomes.
Building on this advantage, counting models for excessive zeros can be further extended by incorporating random effects—referred to here as \textbf{counting mixture models for excessive zeros}—to jointly account for longitudinal dependence and the excess-zeros in claim frequency data \citep{boucher2007risk, boucher2008models, boucher2011correlated, aguero2013full, zhang2022new}.
A seminal discussion of random-effect models and their conceptual link to insurance ratemaking is provided by \citet{nelder1997credibility}, who highlighted the theoretical connection between random effects and credibility theory.

Despite these advancements, an important gap remains.
While counting random-effect models for excessive zeros improve model fit and capture unobserved heterogeneity, their consistency with fundamental probabilistic principles in insurance, particularly \textbf{stochastic monotonicity}, has not been systematically examined.
From an actuarial perspective, it is natural to expect that as the number of past claims increases, the a posteriori claim frequency distribution should indicate a higher level of risk.
\citet{purcaru2003dependence} rigorously established this property for the Poisson random effect model, and also extended the result to the broader exponential dispersion family within generalized linear mixed and dynamic random-effect frameworks.

However, the property of stochastic monotonicity in counting random‐effect models with excessive zeros has not yet been examined in the literature.
Addressing this theoretical gap, we make two main contributions.
First, we provide rigorous mathematical evidence that existing counting random‐effect models for excessive zeros may violate stochastic monotonicity—a property that is not merely undesirable but fundamentally inconsistent with sound actuarial reasoning.
Such violations can distort incentive structures and create potential for moral hazard, leading to unfair or even negative credibility adjustments \citep{pinquet2020poisson, li2021dynamic, ahn2021ordering}.
Second, we propose new classes of counting random‐effect models for excessive zeros that guarantee stochastic monotonicity, ensuring that posterior credibilities consistently increase with claim history and thereby maintaining fairness and theoretical coherence in a posteriori pricing.

%
%We first provides a rigorous mathematical analysis demonstrating how such models can violate stochastic monotonicity, and further introduces new formulations of counting random-effect models for excessive zeros that ensure compliance with this essential property.
%
%	
%This study makes two main contributions.

The remainder of this paper is organized as follows. Section \ref{sec.2} introduces the notations, definitions, and lemmas used throughout the paper. Section \ref{sec.3} analyzes and provides a rigorous mathematical verification that the Poisson-hurdle mixture model with correlated random effects does not, in general, guarantee monotonicity. Section \ref{sec.4} derives and validates conditions under which monotonicity is ensured for Poisson-hurdle mixture models with independent random effects. Section \ref{sec.5} presents a Poisson-hurdle mixture model with comonotonic random effects and derives conditions guaranteeing monotonicity. Extensions to Negative Binomial and zero-inflated specifications are developed in the Appendix.
Section \ref{sec.6} provides a case study using the LGPIF dataset. Finally, Section \ref{sec.8} concludes the paper.
	
\section{Preliminary results}\label{sec.2}
	We introduce the distributions and notations that will be used throughout the paper:
	
	\begin{itemize}
		\item $\operatorname{Ber}(\theta)$: \text{the Bernoulli distribution with success probability} $\theta \in [0, 1] $ .
		\item $\operatorname{Pois}(\lambda)$: Poisson distribution with mean $\lambda>0$.
  \item $\operatorname{NB}(r,p)$: the Negative Binomial distribution with shape parameter
  $r>0$ and success probability $p\in(0,1)$, with probability mass function
  \[
  \P{N=n}
  = \binom{n+r-1}{n} (1-p)^r p^{\,n},
  \qquad n\in\mathbb{Z}_{\ge 0},
  \]
  having mean $\displaystyle \frac{rp}{1-p}$ and variance
  $\displaystyle \frac{rp}{(1-p)^2}$.
		\item $\operatorname{MVN}((\mu_1,\mu_2)^{\top}, \Sigma)$: Bivariate normal distribution with mean vector $ (\mu_1,\mu_2)^{\top} $ and covariance matrix $ \Sigma $.
		\item $\operatorname{Beta}(a,b)$: Beta distribution in (0,1) with mean $\displaystyle\frac{a}{a+b}$  and variance $\displaystyle\frac{ab}{(a+b)^2(a+b+1)}$.
		\item $\operatorname{Gamma}(\alpha,\beta)$: Gamma distribution with shape parameter $ \alpha>0$ and rate parameter $\beta>0 $ with mean $\displaystyle\frac{\alpha}{\beta}$  and variance $\displaystyle\frac{\alpha}{\beta^2}$.
	\end{itemize}

	We use $B(a,b)$ for Beta function, defined as $  B(a, b) = \displaystyle\int_0^1 x^{a - 1}(1 - x)^{b - 1} \, dx$.
We use $\Gamma(\cdot)$ to denote Gamma function, defined as $\Gamma(z) = \displaystyle\int_0^\infty x^{z - 1} e^{-x} \, dx $.
			We use $\sigma(x) = (1+e^{-x})^{-1}$ to denote the logistic function.
We use $x_{1:t}:=(x_1, \cdots, x_t)\in\Real^t$.
We use the notation
\[
\langle x_{1:t}, y_{1:t}\rangle
\]
to denote the inner product between $x_{1:t}, y_{1:t}\in\mathbb{R}^t$.

For a random vector $X_{1:t}$, we write
\[
f(x_{1:t})
\]
for its probability density function when it is absolutely continuous, or for its probability mass function when it is purely discrete.
Unless otherwise specified, we will simply refer to $f$ as the density function, regardless of whether
the underlying distribution is continuous or discrete.
 We also use
\[
f(x_{1:t}\mid y_{1:m})
\]
to denote the conditional density function at $X_{1:t}=x_{1:t}$ conditional on another random vector $Y_{1:m}=y_{1:m}$.

For $x_{1:t}, y_{1:t}\in\mathbb{R}^t$, we write $x_{1:t}\wedge y_{1:t}$ and $x_{1:t}\vee y_{1:t}$ for their componentwise minimum and maximum, respectively, and use $x_{1:t}\le y_{1:t}$ to denote coordinatewise inequality.
A function $g:\mathbb{R}^t \to \mathbb{R}_{\ge 0}$ is called
\textbf{multivariate totally positive of order two (MTP$_2$)} on $D\subseteq \Real^k$ if
\[
g(x_{1:t})\,g(y_{1:t}) \;\ge\; g(x_{1:t}\wedge y_{1:t})\,g(x_{1:t}\vee y_{1:t})
\]
for all $x_{1:t},y_{1:t}\in\mathbb{R}^t$ such that $x_{1:t}\wedge y_{1:t}\in D$ and $x_{1:t}\vee y_{1:t}\in D$.
If $g$ is the density of a random vector $X_{1:t}$, we simply say that $X_{1:t}$ is MTP$_2$.
When $t=2$, the property is referred to as TP$_2$.

%For give random vectors $X_{1:t}$ and $Y_{1:t}$, let $f$ and $g$ denote their density functions.
%We say that $X_{1:t}$ is smaller than $Y_{1:t}$ in the
%\textbf{multivariate likelihood ratio order}\footnote{Jaeyoun: MLR order is not used, hence can be erased?} (MLR order), written $X_{1:t} \le_{\mathrm{LR}} Y_{1:t}$, if
%\[
%f(x_{1:t})\,g(y_{1:t}) \;\le\; f(x_{1:t}\wedge y_{1:t})\,g(x_{1:t}\vee y_{1:t})
%\]
%for all $x_{1:t},y_{1:t}\in\mathbb{R}^t$.

    Let $X$ and $Y$ be univariate random variables with density functions $f$ and $g$, respectively. Then, $X$ is said to be smaller than $Y$ in the \textbf{likelihood ratio order},
    written $X \le_{\mathrm{LR}} Y$, if
    \[
    f(x)g(y)\ge f(y)g(x)
    \]
    for all $x\le y$.
    %Similarly, $X$ is said to be smaller than $Y$ in the \textbf{stochastic order}, written $X \le_{\mathrm{ST}} Y$, if
    %\[
    %\P{X>x}\le \P{Y>x}
    %\]
    %for all $x\in\Real$.

    The followings are well known results.
    %\begin{lemma}[\cite{shaked2007stochastic}]\label{lem.1}
    %Let $X_{1:t}$ and $X_{1:t}^*$ be $t$-dimensional random vectors such that
    %\[
    %X_{1:t} \;\le_{\mathrm{LR}}\; X_{1:t}^*.
    %\]
    %Then $X_{1:t} \le_{\mathrm{st}} X_{1:t}^*$, and in particular
    %\[
    %\E{h(X_{1:t})} \;\le\; \E{h(X_{1:t}^*)}
    %\]
    %  for all non-decreasing function $h:\Real^t\rightarrow \Real$ for which the expectation exists.
    %\end{lemma}

    \begin{lemma}[\cite{denuit2006actuarial}]\label{lem.0}
    Let $\Theta$ be a latent random variable with cumulative distribution function $G$ supported on $\mathcal T\subseteq\mathbb R$.
    Conditional on $\Theta=\theta$, let $X_1,\dots,X_t$ be independent random variables with conditional distribution functions $F_j(\cdot\mid\theta)$ and corresponding densities $f_j(\cdot\mid\theta)$ for $j=1,\dots,t$.
    We further assume that the joint distribution of $X_{1:t}=(X_1,\dots,X_t)$ is given by
    \[
    \P{X_{1:t}\le x_{1:t}}
    = \int_{\mathcal T} \prod_{j=1}^t F_j(x_j\mid\theta)\,{\rm d}G(\theta),
    \qquad x_{1:t}\in\mathbb R^t.
    \]
    If, for each $j$, $f_j(x_j\mid\theta)$ is TP$_2$ in $(x_j,\theta)$, then $X_{1:t}$ is MTP$_2$.
    \end{lemma}

    The following result is a direct consequence of the definition of the MTP$_2$ property.
    Although this fact is well known, we include a brief proof for completeness.
    \begin{lemma}\label{lem.1}
    Assume that $X_{1:t}$ is MTP$_2$. Fix an index $i=1, \cdots, t$ and write $X_{-i}:=(X_1,\dots,X_{i-1},X_{i+1},\dots,X_t)$ and $x_{-i}\in\Real^{t-1}$ denote realizations of $X_{-i}$. Let
    \[
    S:=\left\{ x_{-i}\in\Real^{t-1} \,:\, f(x_{-i})>0\right\}.
    \]
    where $f_{X_{-i}}(x_{-i})$ denotes the density of $X_{-i}$ evaluated at $x_{-i}\in\Real^{t-1}$.
     Then, for any $x_{-i}, x_{-i}'\in S$ with $x_{-i}\le x_{-i}'$, we have
      \[
      \left[X_i\mid X_{-i}=x_{-i} \right]\;\le_{\operatorname{LR}}\;\left[X_i\mid X_{-i}=x_{-i}' \right].
      \]
    \end{lemma}
    \begin{proof}
      It is enough to show
      \begin{equation}\label{eqq1}
      f(x_i' \mid x_{-i}')f(x_i\mid x_{-i})\ge f(x_i\mid x_{-i}')f(x_i'\mid x_{-i}), \qquad \hbox{for any} x_i\le x_i'
      \end{equation}
      for any $x_{-i}, x_{-i}'\in S$ with $x_{-i}\le x_{-i}'$.
      Since \eqref{eqq1} is equivalent with
      \[
      f(x_i', x_{-i}')f(x_i, x_{-i})\ge f(x_i, x_{-i}')f(x_i', x_{-i})
      \]
      we know that the MTP$_2$ property of $X_{1:t}$ implies \eqref{eqq1}.
    \end{proof}

\section{Violation of credibility order in Poisson-hurdle mixture model}\label{sec.3}

Consider the stochastic process $(X_t)_{t\ge 1}$ representing insurance risks with finite means.
In insurance applications, the basic object of interest is the \textbf{base credibility}
\[
\E{X_{t+1}\mid X_{1:t}=x_{1:t}},
\]
which, for brevity, we denote by $\E{X_{t+1}\mid x_{1:t}}$.
From a ratemaking perspective, it is desirable that these conditional expectations preserve the following monotonicity property:
\begin{equation}\label{eq.11}
\E{X_{t+1}\mid x_{1:t}}
\;\le\;
\E{X_{t+1}\mid x_{1:t}'},
\qquad\text{for all } x_{1:t}\le x_{1:t}'.
\end{equation}
Given the diversity of insurance products, it is more natural to require a generalized version of this condition:
\begin{equation}\label{eq.12}
\E{h(X_{t+1})\mid x_{1:t}}
\;\le\;
\E{h(X_{t+1})\mid x_{1:t}'},
\qquad\text{for all } x_{1:t}\le x_{1:t}',
\end{equation}
for every nondecreasing function $h:\mathbb{R}\to\mathbb{R}$ such that the expectation exists.

We say that the process $(X_t)_{t\ge 1}$ possesses the \textbf{base credibility order} if it satisfies~\eqref{eq.11},
and the \textbf{general credibility order}\footnote{This property is referred to in the literature as being
\textbf{conditionally increasing in sequence} \citep{shaked2007stochastic}.} if it satisfies~\eqref{eq.12}.
When no distinction between these two conditions is necessary, we collectively refer to them as the
\textbf{credibility order}.
These ordering properties play a central role in insurance ratemaking, as discussed in \citet{purcaru2003dependence}
and \citet{denuit2006actuarial}.

While the credibility order has been rigorously established for the Poisson random-effects model and, more generally, for generalized linear mixed models (GLMMs)
\citep{purcaru2003dependence, denuit2006actuarial}, no comparable analysis has been conducted for
zero-inflated or hurdle models.
This paper fills that gap by examining mixture models for count data with excessive zeros and analyzing whether they preserve the credibility order property.
In this section, we focus on the Poisson--hurdle mixture model and assess its compliance with credibility ordering.
\ref{app:ext} extends the analysis to hurdle mixture models with more general counting distributions, and also to zero-inflated specifications.
We begin by formulating a general version of the Poisson--hurdle mixture model, which serves as the foundation for the model variants studied later in the paper.

%
%Before that this section focuses on the Poisson-hurdle mixture models, later in Section XXX zero-inflated mixture models also, and show that the credibility can be violated.
%
%shows that the credibility orders of mixture models for count data can be violated.
%
%
%\citet{purcaru2003dependence} shows that credibility order is well preserved for the Poisson random effect model, and
%\citet{denuit2006actuarial} extended this result to the broader exponential dispersion family within generalized linear mixed and dynamic random-effect frameworks
%
%
%
%
%, more specifically the Poisson-hurdle mixture models with .
%
%This paper examines mixture models for count data exhibiting an excess of zeros, and investigates
%the credibility order of these models.
%Two principal approaches are available for modeling count data with excessive zeros:
%the \textit{zero-inflated} models
%\citep{lambert1992zero, yip2005modeling, boucher2009number} and the \textit{hurdle} models
%\citep{boucher2008modelling, zhang2022new}.
%Mathematically, these two frameworks are closely related, although their interpretations differ.
%

	\begin{model}\label{mod.0}
		The \textbf{Poisson–hurdle mixture model with general random effects} for $(Y_{t})_{t\ge 1}$ is defined as
		\[
		Y_{t} := Z_{t}(1+N_{t}), \quad t \ge 1,
		\]
		under the following assumptions:
		
		\begin{itemize}
\item[i.] \textbf{Observation distribution.}
Conditional on $\Theta=(\Theta_1,\Theta_2)$, the pairs $(Z_t,N_t)$ are independent \ across $t$, and
$Z_t \perp N_t \mid \Theta$.
Their conditional distributions are
\begin{equation}\label{eq.d.1}
Z_t \mid \Theta \;\sim\; \mathrm{Ber}\!\big(\eta_1(\Theta_1)\big),
\qquad
N_t \mid \Theta \;\sim\; \mathrm{Pois}\!\big(\eta_2(\Theta_2)\big),
\end{equation}
where $\eta_1: \mathcal{R}(\Theta_1) \to (0,1)$ and $\eta_2: \mathcal{R}(\Theta_2) \to (0,\infty)$ are proper activation functions for Bernoulli distribution and Poisson distribution, respectively. Here, $\mathcal{R}(\Theta_j)$ denotes the set of possible values taken by $\Theta_j$ $(j=1,2)$.

\item[ii.] \textbf{Random effects distribution.}
The random effects $\Theta=(\Theta_1,\Theta_2)$ follow a general joint distribution
$\mathcal{L}(\Theta)$ supported on
$\mathcal{R}(\Theta_1)\times \mathcal{R}(\Theta_2)$.

		\end{itemize}
		
	\end{model}

\subsection{Limitations of the Poisson--hurdle mixture model with correlated random effects}

Here, we consider a specific case of Model~\ref{mod.0}, namely the Poisson--hurdle mixture model with correlated random effects, where the random effects follow a bivariate normal distribution.
We examine its limitations with respect to the credibility order property, and Section~\ref{sec.3.2} presents numerical examples that clearly demonstrate violations of these orders.
The formal specification of this model is provided below.

	\begin{model}\label{mod.1}
The \textbf{Poisson–hurdle mixture model with bivariate normal random effects}
for $(Y_t)_{t\ge 1}$ is given by Model~\ref{mod.0} with the following specifics:
		\begin{itemize}
\item[i.] \textbf{Observation distribution.} The conditional distributions in \eqref{eq.d.1} is specified by
\[
Z_t \mid \Theta \;\sim\; \mathrm{Ber}\!\big(\sigma(\Theta_1)\big),
\qquad
N_t \mid \Theta \;\sim\; \mathrm{Pois}\!\big({\mathrm {exp}}(\Theta_2)\big).
\]
			\item[ii.] \textbf{Random effects distribution.}
			The random effects follow a bivariate normal distribution,
			\[
			\Theta \;\sim\; \operatorname{MVN}\!\left(\begin{pmatrix}
  \mu_1\\
  \mu_2
\end{pmatrix}, \Sigma\right),
			\]
			with positive definite covariance matrix
			\[
			\Sigma =
			\begin{bmatrix}
				\sigma_1^2 & \rho\,\sigma_1\sigma_2 \\
				\rho\,\sigma_1\sigma_2 & \sigma_2^2
			\end{bmatrix}
			\]
for $\sigma_1, \sigma_2>0$ and $\rho\in(-1,1)$.
		\end{itemize}
		
	\end{model}

Our first interest lies in to investigate whether Model~\ref{mod.1} preserve the most basic form of the base credibility order; the inequality
	\[
	\E{Y_{2}\mid Y_{1}=y_1} \;\le\; \E{Y_{2}\mid Y_{1}=y_1'}
	\quad\text{for } y_1 \le y_1'.
	\]
Noting that $Y_1=0$ is equivalent with $Z_1=0$, and $Y_1=y_1$ is equivalent with $Z_1=1$ and $N_1=y_1-1$ for $y_1=\mathbb{Z}_{\ge 1}$, we distinguish two different cases
\begin{itemize}
  \item[i.] For $y_1\in\mathbb{Z}_{\ge 1}$, we are interested in the following inequality
  \begin{equation}\label{eq.310}
  \E{Y_{2}\mid Y_{1}=y_1+1} \;\ge\; \E{Y_{2}\mid Y_{1}=y_1}
  \end{equation}
  which is equivalent with
  \begin{equation}\label{eq.311}
  \E{Y_{2}\mid Z_{1}=1, N_1=y_1} \;\ge\; \E{Y_{2}\mid Z_1=1, N_1=y_1-1}.
  \end{equation}
  \item[ii.] We are interested in the following inequality
  \begin{equation}\label{eq.321}
  \E{Y_{2}\mid Y_{1}=1} \;\ge\; \E{Y_{2}\mid Y_{1}=0}
  \end{equation}
  which is equivalent with
  \begin{equation}\label{eq.322}
  \E{Y_{2}\mid Z_{1}=1, N_1=0} \;\ge\; \E{Y_{2}\mid Z_1=0}.
  \end{equation}
\end{itemize}

In comparing the inequality in \eqref{eq.311}, the main difference between the two terms lies in $N_1 = y_1$ versus $N_1 = y_1 - 1$.
If we assume that the magnitude of $N_1$ has a certain \emph{positive influence}\footnote{The precise meaning of ``positive influence'' will be clarified in the next lemma.} on $Y_2$, then the inequality in \eqref{eq.311} is expected to hold, as formally established in Lemma~\ref{lem.2}. The proof of Lemma \ref{lem.2} is in \ref{app:proofs}.

%
%First, the following lemma, whose proof can be found in \ref{app:proofs}, shows that within the positive-count stratum ($Z_1 = 1$), the monotonicity condition in \eqref{eq.310} holds under a positive dependence between the random effects.
%

	\begin{lemma}\label{lem.2}
		Under Model~\ref{mod.1}, we further assume $\rho \in [0,1)$ and $\sigma_1^2,\sigma_2^2>0$. Then, for any integers \(1 \le y_1 \le y_1'\), we have the following
\begin{itemize}
  \item[i.] We have the following likelihood ratio order
  \[
  \left[\Theta \mid y_1\right] \le_{\rm LR} \left[\Theta \mid y_1'\right].
  \]
  \item[ii.] For all non-decreasing function $h:\Real_{\ge 0} \rightarrow \Real$ for which the expectation exists, we have
  \[
		\E{h(Y_{2})\mid Y_{1}=y_1}
		\;\le\;
		\E{h(Y_{2})\mid Y_{1}=y_1'}.
		\]
\end{itemize}

	\end{lemma}

In contrast, the comparison in \eqref{eq.322} involves $(Z_1, N_1) = (1, 0)$ versus $Z_1 = 0$.
Even under the same assumption that both $Z_1$ and $N_1$ have positive influence on $Y_2$, it is unclear whether the inequality in \eqref{eq.322} holds.
This ambiguity arises because the information $Z_1 = 1$ (on the left-hand side) versus $Z_1 = 0$ (on the right-hand side) supports the inequality in \eqref{eq.322}, whereas the information $N_1 = 0$ (on the left) versus the absence of information on $N_1$ (on the right) suggests the opposite direction.
Hence, it remains indeterminate which effect dominates. Indeed, Theorem~\ref{thm.1} shows that, under suitable assumptions on the random–effects distribution with positive correlation, the inequality in~\eqref{eq.321} can be reversed, thereby violating the base credibility order.

%We begin by comparing the base credibilities corresponding to the two past outcomes
%\(Y_1=0\) and \(Y_1=1\).
%The following theorem, proved in Appendix~\ref{app:proofs}, shows that the inequality in~\eqref{eq.321} may fail to hold, and hence consequently the base credibility order can be violated.
%

%Next, in the cross-hurdle comparison ($Y_1=0$ versus $Y_1=1$), the following theorem, whose proof can be found in \ref{app:proofs}, shows that the inequality in~\eqref{eq.321} can be reversed, thereby violating the base credibility order.

\begin{theorem}\label{thm.1}
Consider the setting of Model~\ref{mod.1}.
For fixed $\sigma_2>0$ and $\rho\in(-1,1)$, we have
\[
   \lim_{\sigma_1^2\to 0^+}
   \left(\,\E{Y_{2}\mid Y_{1}=0}
           -\E{Y_{2}\mid Y_{1}=1}\right)
   \;>\;0.
\]
\end{theorem}

	From both a regulatory and actuarial perspective, the violation of credibility order is problematic. Such violations can lead to undesirable incentives, such as the strategic submission of minor claims with the aim of minimizing future credibility costs, resembling arbitrage situations, or otherwise undermine fairness in a posteriori pricing. In practice, this can foster moral hazard in which lower–risk policyholders file trivial/fake claims to avoid potential credibility increases.

	\subsection{Numerical examples}\label{sec.3.2}

We illustrate several numerical examples within the framework of Model~\ref{mod.1} in which
the base credibility order is violated.
Our analysis primarily focuses on the comparison of
$\mathbb{E}[Y_2 \mid Y_1=y_1]$ for $y_1=0$ and $y_1=1$,
since comparisons involving $y_1 \ge 1$ are guaranteed to be ordered under the conditions of Lemma~\ref{lem.2}.
As shown in \ref{app:Example}, we also note that this ordering for $y_1 \ge 1$ need not hold when the correlation parameter $\rho$ is negative.
Since Model~\ref{mod.1} does not admit a closed-form posterior distribution, we approximate the
relevant quantities using Markov chain Monte Carlo (MCMC) with a Monte Carlo sample size of
$S=1000$. To assess the Monte Carlo variability and compute standard errors, the estimation procedure
is repeated over $100$ independent runs.
% {\color{red}(Jaeyoun to Hyemin: when we calculate the standard error, simulations are repeteated. We have not specify this number of repetition.)}

Table~\ref{sigma1_simulation} reports Monte Carlo estimates of
$\mathbb{E}[Y_{2}\mid Y_{1}=0]$ and $\mathbb{E}[Y_{2}\mid Y_{1}=1]$
obtained by varying $\sigma_1^2$ while fixing $\mu_1=\mu_2=0$, $\sigma_2^2=1$, and $\rho=0.5$.
As $\sigma_1^2$ decreases, the ordering between the two conditional expectations is reversed, namely,
\[
\E{Y_2 \mid Y_1=1}\le \E{Y_2 \mid Y_1=0}.
\]
This behavior is consistent with Theorem~\ref{thm.1}.

%First, we examine how the conditional expectations
%$\mathbb{E}[Y_2 \mid Y_1=0]$ and $\mathbb{E}[Y_2 \mid Y_1=1]$ evolve as the parameter
%$\sigma_1^2$ approaches zero, thereby supporting Theorem~\ref{thm.1}.
%
%
%In addition, through further simulation studies, we show that violations of the base credibility order can also occur when $\sigma_2^2$ or $\mu_2$ takes sufficiently large values.

% }

Table~\ref{sigma_2_simulation} presents Monte Carlo estimates obtained by varying $\sigma_{2}^2$
while fixing $\mu_1=\mu_2=0$, $\sigma_{1}^2=1$, and $\rho=0.5$.
As $\sigma_{2}^2$ increases, violations of the base credibility order are again observed.

\begin{table}[ht]
\centering
\caption{Comparison of
$\E{Y_2 \mid Y_1=0}$ and $\E{Y_2 \mid Y_1=1}$ as $\sigma_{1}^2$ decreases}
\label{sigma1_simulation}
\begin{tabular}{ccc}
\toprule
$\sigma_{1}^2$ & $\mathbb{E}[Y_2 \mid Y_1 = 0]$ & $\mathbb{E}[Y_2 \mid Y_1 = 1]$ \\
\midrule
$5.0$ & 0.7113 (0.0085) & 1.2061 (0.0064) \\
$2.0$ & 0.9319 (0.0101) & 1.0576 (0.0052) \\
$1.0$ & 1.0780 (0.0091) & 0.9630 (0.0046) \\
$0.1$ & 1.3073 (0.0110) & 0.8396 (0.0032) \\
\bottomrule
\end{tabular}
\caption*{\footnotesize Note: Parentheses report Monte Carlo standard errors (MCSEs)}
\end{table}

%Table~\ref{sigma_2_simulation} reports Monte Carlo estimates ($S=300{,}000$) obtained by varying $\sigma_{2}^2$ while fixing $\mu_1=\mu_2=0$, $\sigma_{1}^2=1$, and $\rho=0.5$.
%As $\sigma_{2}^2$ increases-that is, as the variance of the random effect in the post-hurdle stage becomes larger-we can confirm that violations of the base credibility order occur.

\begin{table}[ht]
\centering
\caption{Comparison of
$\E{Y_2 \mid Y_1=0}$ and $\E{Y_2 \mid Y_1=1}$ as $\sigma_{2}^2$ increases}
\label{sigma_2_simulation}
\begin{tabular}{ccc}
\toprule
$\sigma_{2}^2$ & $\mathbb{E}[Y_2 \mid Y_1 = 0]$ & $\mathbb{E}[Y_2 \mid Y_1 = 1]$ \\
\midrule
0.01 & 0.8269 (0.0036) & 1.1732 (0.0034) \\
0.10 & 0.8436 (0.0040) & 1.1539 (0.0037) \\
1.00 & 1.0780 (0.0091) & 0.9630 (0.0046) \\
2.00 & 1.5067 (0.0309) & 0.8617 (0.0042) \\
\bottomrule
\end{tabular}
\caption*{\footnotesize Note: Parentheses report MCSEs.}
\end{table}

Moreover, beyond variations in the covariance parameters $\sigma_1^2$ and $\sigma_2^2$,
violations of the base credibility order may also arise when $\mu_2$ is sufficiently large.
Table~\ref{mu2_simulation} reports Monte Carlo estimates obtained by varying $\mu_2$
while fixing $\mu_1=0$, $\sigma_1^2=\sigma_2^2=1$, and $\rho=0.5$.
As $\mu_2$ increases, violations of the base credibility order are observed.

\begin{table}[ht]
\centering
\caption{Comparison of
$\E{Y_2 \mid Y_1=0}$ and $\E{Y_2 \mid Y_1=1}$ as $\mu_{2}$ increases}
\label{mu2_simulation}
\begin{tabular}{ccc}
\toprule
$\mu_{2}$ & $\mathbb{E}[Y_2 \mid Y_1 = 0]$ & $\mathbb{E}[Y_2 \mid Y_1 = 1]$ \\
\midrule
$-2$ & 0.4976 (0.0026) & 0.6969 (0.0026) \\
$-1$ & 0.6526 (0.0043) & 0.8061 (0.0031) \\
$0$  & 1.0797 (0.0099) & 0.9554 (0.0056) \\
$1$  & 2.1982 (0.0269) & 1.1227 (0.0060) \\
$2$  & 5.2635 (0.0705) & 1.2502 (0.0068) \\
\bottomrule
\end{tabular}
\caption*{\footnotesize Note:Parentheses report MCSEs.}
\end{table}

% \begin{remark}
%   \color{blue}In case of Zero-inflated model which is the version of Model \ref{mod.1}, the violation of the base credibility order is similarly observed, and we provide the specific numerical example in Appendix XXX.
% \end{remark}

\begin{remark}\label{rmk.1}
The same phenomenon can occur for the \emph{zero-inflated} counterpart of Model~\ref{mod.1}, see Model \ref{mod.7}. In particular, violations of the base credibility order are demonstrated in \ref{app:Example}.
\end{remark}

\section{Poisson–hurdle mixture model with the independent random effects}\label{sec.4}

%When claim histories are restricted to positive counts, Lemma~\ref{lem.2} ensures the
%general credibility order under positive correlation of the random effects. By contrast,
Theorem~\ref{thm.1}, together with the numerical examples, indicates that there is no simple or
universal joint distribution of the random effects that guarantees monotonicity in the
Poisson--hurdle mixture model.
Motivated by this challenge, we now consider tractable subclasses of dependence structures that
facilitate analytical derivations of the conditions under which the credibility order holds.
Specifically, we examine two extremal cases of dependence between the random effects:
independent random effects (this section) and comonotonic random effects (Section~\ref{sec.5}).
These two cases provide mathematically convenient frameworks that simplify the analysis while
allowing for a rigorous investigation of credibility orders.

In this section, we introduce a Poisson--hurdle mixture model with independent random effects.
Unlike Model~\ref{mod.1}, the random effects governing the hurdle and count components are assumed
independent and are specified via conjugate priors: a Beta distribution for the Bernoulli hurdle
part and a Gamma distribution for the Poisson positive-count part. This formulation yields a
tractable framework with closed-form predictive distributions, providing a convenient setting for
examining conditions under which the base credibility order is satisfied.
It should be emphasized, however, that this model satisfies only the \emph{base credibility order} and
not the \emph{general credibility order}, as demonstrated in Example~\ref{ex.1} below.
Note that Section~\ref{sec.5} presents the comonotonic random–effects model, which achieves the general credibility order.
Accordingly, the model developed in this section should be regarded as a \emph{motivating example} that provides a preliminary illustration of the main ideas.

	\begin{model}\label{mod.2}
The \textbf{Poisson–hurdle mixture model with independent random effects}
for $(Y_t)_{t\ge 1}$ is given by Model~\ref{mod.0} with the following specifics:

		\begin{itemize}
\item[i.] \textbf{Observation distribution.} The conditional distributions in \eqref{eq.d.1} is specified by
\[
Z_t \mid \Theta \;\sim\; \mathrm{Ber}\!\big(\Theta_1\big),
\qquad
N_t \mid \Theta \;\sim\; \mathrm{Pois}\!\big(\Theta_2\big).
\]
			\item[ii.] \textbf{Random effects distribution.}
The random effects $\Theta_{1}$ and $\Theta_{2}$ are independent with the following specifications
			\[
			\Theta_{1} \sim \operatorname{Beta}\!\bigl(a,\,b\bigr),
			\qquad
			\Theta_{2} \sim \operatorname{Gamma}\!\bigl(\alpha,\,\beta\bigr),
			\]
			for some constants $a, b, \alpha, \beta>0$.
		\end{itemize}

	\end{model}

% \begin{remark}{\color{red}[Jaeyoun: This remark will be moved to data analysis part.]}
% Due to convenient conjugacy property, the Poisson–Gamma random–effects model is widely used in insurance applications.
% In this context, the a priori rate, specific to each policyholder, is employed to capture observable heterogeneity through explanatory covariates (not to be confused with the unobserved heterogeneity modeled via random effects).
% In the standard Poisson–Gamma formulation, the canonical specification is
% \[
% N_{i,t}\,\overset{\mathrm{iid}}{\sim}\,\operatorname{Pois}(\lambda_i\,\Theta^{[2]}_i),
% \qquad
% \Theta^{[2]}_i\sim\operatorname{Gamma}(\alpha,\beta).
% \]
% Since $\lambda_i\Theta^{[2]}_i\sim \operatorname{Gamma}\!\bigl(\alpha,\,\beta/\lambda_i\bigr)$, this specification is equivalent to
% \[
% N_{i,t}\,\overset{\mathrm{iid}}{\sim}\,\operatorname{Pois}(\Theta_{i,2}^*),
% \qquad
% \Theta_{i,2}^*\sim\operatorname{Gamma}(\alpha,\beta/\lambda_i).
% \]
% Thus, the covariates enter through the rate parameter of the Gamma distribution.
% By contrast, in Model~\ref{mod.2} we attach covariates to the \emph{shape} parameter $\alpha_i$ (and to the Beta prior via $b_i$).
% This choice greatly facilitates the monotonicity analysis of the predictive mean (see Section~4.1).
% \end{remark}

The conjugate structure of the Bernoulli–Beta and Poisson–Gamma families ensures that Model~\ref{mod.2} possesses the following closed-form posterior distribution.

\begin{corollary}\label{prop.1}
Consider the setting of Model~\ref{mod.2}. For the claim history $Y_{1:t}=y_{1:t}$, define
\[
r_t := \sum_{s=1}^t \mathbf{1}_{\{y_s>0\}},
\qquad\hbox{and}\quad
m_t := \sum_{s=1}^t (y_s-1)\,\mathbf{1}_{\{y_s>0\}}.
\]
where $\sum_{s=1}^{0}$ is defined to be zero.
Then the posterior of $\Theta=(\Theta_1,\Theta_2)$ is given by
\[
\pi(\theta_1,\theta_2\mid y_{1:t})
=\frac{1}{B(a^*_t,b^*_t)}\,(\theta_1)^{a^*_t-1}(1-\theta_1)^{b^*_t-1}
\cdot
\frac{(\beta^*_t)^{\alpha^*_t}}{\Gamma(\alpha^*_t)}\,(\theta_2)^{\alpha^*_t-1}e^{-\beta^*_t\theta_2},
\]
with updated parameters
\[
a^*_t = a + r_t,\quad
b^*_t = b + t - r_t,\quad
\alpha^*_t = \alpha + m_t,\quad
\beta^*_t = \beta + r_t.
\]
Consequently, the one–step–ahead predictive mean is
\[
\E{Y_{t+1}\mid y_{1:t}}
= \frac{a+r_t}{a+b+t}\;\Biggl(1+\frac{\alpha+m_t}{\beta+r_t}\Biggr).
\]

\noindent In particular, after a single observation $Y_{1}=y_{1}$,
\begin{equation}\label{eq:BGoneahead}
\E{Y_{2}\mid Y_{1}=y_{1}}
=
\begin{cases}
\dfrac{a}{a+b+1}\,\Bigl(1+\dfrac{\alpha}{\beta}\Bigr), & y_{1}=0,\\[10pt]
\dfrac{a+1}{a+b+1}\,\Bigl(1+\dfrac{\alpha+(y_{1}-1)}{\beta+1}\Bigr), & y_{1}>0.
\end{cases}
\end{equation}
\end{corollary}
\begin{proof}
  The proof follows from Bernoulli-Beta and Poisson-Gamma conjugacies.
\end{proof}

From Corollary~\ref{prop.1}, we obtain the following result, which provides the necessary and sufficient condition for the base credibility order.

\begin{proposition}\label{prop.2}
The process $\left(Y_t\right)_{t\ge 1}$ in Model~\ref{mod.2} satisfies the base credibility order; i.e.,
		\begin{equation}\label{eq.5}
			\E{Y_{t+1} \mid Y_{1:t}=y_{1:t}}
			\le
			\E{Y_{t+1} \mid Y_{1:t}=y_{1:t}'}, \qquad \forall y_{1:t}\le y_{1:t}'
		\end{equation}
for all $t\ge 1$
 if and only if
 \begin{equation}\label{eq.53}
			\alpha_t^*\,a_t^* \;\le\; \beta_t^*\,(\alpha_t^*+\beta_t^*+1)
\end{equation}
for all $t\ge 1$ (previously it was $t\ge 0$ please check this change is correct), where
\(
\alpha_{t}^*, \beta_{t}^*, a_{t}^*, b_{t}^*>0
\)
are defined in Corollary \ref{prop.1}.
\end{proposition}
\begin{proof}
Fix $t\ge 0$.
By Corollary~\ref{prop.1}, for a history $y_{1:t}$ the one–step–ahead predictive mean can be written as
\[
\Phi_t(r_t,m_t)
\;:=\;
\E{Y_{t+1}\mid y_{1:t}}
\;=\;
\frac{a+r_t}{a+b+t}\left(1+\frac{\alpha+m_t}{\beta+r_t}\right),
\]
where non-negative integers $r_t$ and $m_t$ are defined in Corollary \ref{prop.1}.

For fixed non-negative integer $r$, we have
\[
\Phi_t(r_t,m_t+1)-\Phi_t(r_t,m_t)
%=\frac{a+r}{a+b+t}\left(\frac{\alpha+m+1}{B_t}-\frac{\alpha+m}{B_t}\right)
=\frac{a+r_t}{(a+b+t)\,(\beta+r_t)}\;>\;0.
\]
Thus $\Phi_t$ is (strictly) increasing in non-negative integers $m_t$.

Next, for fixed non-negative integer $m_t$, we have
\begin{equation}\label{eq.p3}
\Phi_t(r_t+1,m_t)-\Phi_t(r_t,m_t)
=
\frac{(\beta+r_t)^2+\beta+r_t+(\alpha+m_t)(\beta+r_t)-(\alpha+m_t) (a+r_t)}{(a+b+t)\,(\beta+r_t)(\beta+r_t+1)}.
\end{equation}
Since the denominator is positive, this difference is nonnegative if and only if
\begin{equation}\label{eq.p4}
(\alpha+m_t) (a+r_t) \;\le\; (\beta+r_t)((\alpha+m_t)+(\beta+r_t)+1)
%\quad\text{i.e.}\quad
%\alpha_t^*\,a_t^* \;\le\; \beta_t^*\,(\alpha_t^*+\beta_t^*+1).
\end{equation}
which is exactly the inequality in \eqref{eq.53}.

Hence, if we suppose that inequality in \eqref{eq.53} holds for all $t\ge 1$, then $\Phi_t$ is nondecreasing in both arguments.
Because $y_{1:t}\le y'_{1:t}$ implies $r_t\le r_t'$ and $m_t\le m_t'$, we obtain
\[
\E{Y_{t+1}\mid y_{1:t}}=\Phi_t(r_t,m_t)\;\le\;\Phi_t(r_t',m_t')=\E{Y_{t+1}\mid y'_{1:t}},
\]
which is \eqref{eq.5}.

Conversely, suppose \eqref{eq.5} holds for all $t$ and all histories.
For fixed $t\ge 1$, consider the claim history $Y_{1:t}=y_{1:t}$ with at least one of $y_1, \cdots, y_t$ are zero.
% any feasible pair of non-negative integers, $(r_t,m_t)$ for given non-negative integers $Y_{1:t}=y_{1:t}$.
Choose another claim history that are identical with $Y_{1:t}=y_{1:t}$ except that one component changes from $0$ to $1$; then $r_t$ increases to $r_t+1$ while $m_t$ is unchanged.
Applying \eqref{eq.5} to this pair yields $\Phi_t(r_t,m_t)\le \Phi_t(r_t{+}1,m_t)$, which, by the equivalence between in \eqref{eq.p3} and \eqref{eq.p4}, is equivalent to
\[
\alpha_t^*\,a_t^* \;\le\; \beta_t^*\,(\alpha_t^*+\beta_t^*+1).
\]
As $(r_t,m_t)$, hence $(\alpha_t^*,\beta_t^*,a_t^*)$, and $t$ were arbitrary, the condition must hold for all $t\ge 1$.
\end{proof}

While Model~\ref{mod.2} yields convenient closed-form posteriors, the base credibility order does not hold universally;
it is governed precisely by the necessary and sufficient condition~\eqref{eq.53} in Proposition~\ref{prop.2}.
Since this condition must be verified for all $t\ge 1$, it may be cumbersome to check in practice.
To overcome this difficulty, the following corollary presents a sufficient condition that is simple to verify.
Moreover, this condition can be imposed directly at the model specification stage, ensuring that the resulting model
automatically satisfies the base credibility order.

	\begin{corollary}\label{thm.3} Consider the setting in Model~\ref{mod.2}.  If
		\begin{equation}\label{ass.1}
			a < \beta,
		\end{equation}
		then, for any \(t \in \mathbb{Z}_{\ge 1}\) and any histories \(y_{1:t},\,y'_{1:t} \in \mathbb{Z}_{\ge 0}^t\) with \(y_{1:t} \le y'_{1:t}\), we have
		\begin{equation}\label{eq.51}
			\E{Y_{t+1}\mid Y_{1:t}=y_{1:t}}
			\;\le\;
			\E{Y_{t+1}\mid Y_{1:t}=y_{1:t}'}.
		\end{equation}
	\end{corollary}
	\begin{proof}
First, note that the condition in \eqref{ass.1} implies
\[
			\alpha_{t}^*\,a_{t}^* \;\le\; \beta_{t}^*\,(\alpha_{t}^*+\beta_{t}^*+1), \quad \hbox{for all}\quad t\ge 1
\]
where
\[
\alpha_{t}^*,\quad \beta_{t}^*,\quad a_{t}^*,\quad b_{t}^*
\]
are defined in Corollary \ref{prop.1}. Then Proposition \ref{prop.2} concludes the proof.

	\end{proof}

While the condition in \eqref{eq.53} or \eqref{ass.1} guarantees the base credibility order,
it does not ensure that the general credibility order is satisfied.
The following numerical example demonstrates that the general credibility order can indeed be violated even under these conditions.
\begin{example}\label{ex.1}
Under Model~\ref{mod.2}, we fix $a=0.5$, $b=1$, $\alpha=1$, and $\beta=1$, so that the base credibility order is satisfied.
We compare the deductible-adjusted conditional expectations
\[
\E{(Y_2-d)^+\,\middle|\,Y_1=0}
\quad\text{and}\quad
\E{(Y_2-d)^+\,\middle|\,Y_1=1}
\]
for various values of the deductible $d$, as reported in Table~\ref{table_Model3_MGCO}.
The results show that, as $d$ increases, violations of the general credibility order become more pronounced.
\end{example}

\begin{table}[h!]
\centering
\caption{Comparison of $\mathbb{E}\!\left[(Y_2-d)^+\,\middle|\,Y_1\right]$ across deductibles $d$}
\label{table_Model3_MGCO}
\begin{tabular}{rccccccccc}
\toprule
 & \multicolumn{9}{c}{d} \\
\cmidrule(lr){2-10}
  & 1 & 2 & 3 & 4 & 5 & 6 & 7 & 8 & 9 \\
\midrule
$\mathbb{E}\!\left[(Y_2-d)^+\,\middle|\,Y_1{=}0\right]$ & 0.200 & 0.100 & 0.050 & 0.025 & 0.013 & 0.006 & 0.003 & 0.002 & 0.001 \\
$\mathbb{E}\!\left[(Y_2-d)^+\,\middle|\,Y_1{=}1\right]$ & 0.300 & 0.100 & 0.033 & 0.011 & 0.004 & 0.001 & 0.000 & 0.000 & 0.000 \\
\bottomrule
\end{tabular}
\end{table}

	% \begin{example}\label{ex.1}
	% 	{\color{red}(Jaeyoun: Please expand the setting of this example as in Table \ref{tab:sigma1} for various settings, and show that some orders of deductible/stop-loss insurance can be violated. This need to be in-line with the new result of deductible orders to be added by Hyemin.)}
		
	% \end{example}

\section{Poisson–hurdle mixture model with comonotonic random effects}\label{sec.5}

In the previous section we showed that under Model~\ref{mod.2}, which assumes independent random effects, the
base credibility order in \eqref{eq.11} may hold depending on the specification of the random‐effects
distributions, whereas the general credibility order in \eqref{eq.12} does not hold in general.
We now consider a Poisson--hurdle mixture with comonotonic random effects.
Analogous to the simplification afforded by independence in Model~\ref{mod.2},
comonotonicity provides a simple dependence structure that enables easier analysis of credibility order.
We begin with the model specification under comonotonic random effects.

	\begin{model}\label{mod.05}
For the given exogenous
\[
(c_t)_{t\ge 1}\subset \Real \quad\hbox{and}\quad (d_t)_{t\ge 1}\subset \Real,
\]
the \textbf{Poisson–hurdle mixture model with comonotonic random effects}
for $(Y_t)_{t\ge 1}$ is given by Model~\ref{mod.0} with the following specifics:
		\begin{itemize}
\item[i.] \textbf{Observation distribution.} The conditional distributions in \eqref{eq.d.1} is specified by
\[
			Z_{t}\mid\Theta \;\sim\; \operatorname{Ber}\!\left(\sigma(c_t+\Theta)\right),
\qquad
			N_{t}\mid\Theta \;\sim\; \operatorname{Pois}\!\left(\lambda\left(d_t+\Theta\right)\right),
\]
for a differentiable, nondecreasing function $\lambda:\Real\to\Real_{>0}$.

\item[ii.] \textbf{Random effects distribution.}
Let $\Theta$ be a real-valued random variable with distribution $\mathcal{L}(\Theta)$
supported on $\mathcal{R}(\Theta)\subseteq\mathbb{R}$.

%
%\item[ii.] \textbf{Random effects distribution.}
%Assume
%\[
%\Theta\sim \operatorname{N}(0, \kappa^2).
%\]
%for some constant $\kappa^2>0$.
		\end{itemize}

	\end{model}

\begin{lemma}\label{lem:TP2-time-s}
Consider the setting of  Model~\ref{mod.05}.
For a fixed time $t\in\mathbb{Z}_{\ge 1}$, define
\[
\eta_t(y,\theta)
:=\mathbb P(Y_t=y\mid \theta)
=
\begin{cases}
1-\sigma(c_t+\theta), & y=0;\\
\sigma(c_t+\theta)\,e^{- \lambda(d_t+\theta)}
\dfrac{\big(\lambda(d_t+\theta)\big)^{\,y-1}}{(y-1)!}, & y\ge 1;
\end{cases}
\]
where
$\lambda(\theta):=e^{\psi(\theta)}>0$.
Then $\eta_t$ is TP$_2$ in $(y,\theta)$ if and only if
\[
0\;\le\;\frac{\partial}{\partial \theta} \lambda(d_t+\theta)\;\le\;1
\]
for $d_t\in\Real$ and all $\theta$ on the support of $\Theta$.
%Then $\eta_t$ is TP$_2$ in $(y,\theta)$.
\end{lemma}

\begin{proof}
To show that $\eta_t(y, \theta)$ is TP$_2$ in $(y, \theta)$, it is enough to show
\[
\frac{\eta_t(y+1, \theta)}{\eta_t(y, \theta)}
\]
are non-decreasing function of $\theta\in\mathcal{R}(\Theta)$ for all $y\in\mathbb{Z}_{\ge 0}$.

When $y=0$, we have
\[
\frac{\eta_t(y+1, \theta)}{\eta_t(y, \theta)}={\rm exp}\left(c_t + \theta  - \lambda(d_t+\theta) \right)
\]
and
\begin{equation}\label{eq.pf.1}
\frac{\partial}{\partial \theta} \frac{\eta_t(y+1, \theta)}{\eta_t(y, \theta)} ={\rm exp}\left(c_t + \theta  - \lambda(d_t+ \theta) \right)\left( 1-\frac{\partial}{\partial \theta} \lambda(d_t+\theta)\right).
\end{equation}
When $y\ge 1$, we have
\[
\frac{\eta_t(y+1, \theta)}{\eta_t(y, \theta)}=\frac{\lambda(d_t+\theta)}{y}
\]
and
\begin{equation}\label{eq.pf.2}
\frac{\partial}{\partial \theta} \frac{\eta_t(y+1, \theta)}{\eta_t(y, \theta)} = \frac{\frac{\partial}{\partial \theta} \lambda(d_t+\theta)}{y}.
\end{equation}

%From the following observation
%\[
%\frac{\partial}{\partial \theta} \lambda(d_t+\theta)=\frac{e^{d_t+\theta}}{1+e^{d_t+\theta}}\in(0,1)
%\]
Finally combining the inequality in \eqref{eq.pf.1} and \eqref{eq.pf.2}, we have
\[
\frac{\partial}{\partial \theta} \frac{\eta_t(y+1, \theta)}{\eta_t(y, \theta)}\ge 0
\]
for all $y\in\mathbb{Z}_{\ge 0}$ and all $\theta\in\mathcal{R}(\Theta)$ if and only if
\[
0\;\le\;\frac{\partial}{\partial \theta_1} \lambda(d_t+\theta)\;\le\;1
\]
for $d_t\in\Real$ and all $\theta\in\mathcal{R}(\Theta)$.
\end{proof}

\begin{theorem}\label{thm.2}
Consider the setting of  Model~\ref{mod.05}.
For given $d_t\in\Real$, suppose
\begin{equation}\label{eq:scaled-Lip}
0\;\le\;\frac{\partial}{\partial \theta} \lambda(d_t+\theta)\;\le\;1, \qquad \forall\theta\in\mathcal{R}(\Theta).
\end{equation}
%Assume
%\[
%\mu_{\max}:=\sup_s \mu_s<\infty.
%\]
%Suppose that on the support of $\Theta_1$ we have
%\begin{equation}\label{eq:scaled-Lip}
%0\;\le\;\mu_{\max}\,\frac{\partial}{\partial \theta_1} \lambda(\theta_1)\;\le\;1
%\end{equation}
%where $\lambda(\theta_1):=e^{\psi(\theta_1)}>0$.
Then, for any \(t \in \mathbb{Z}_{\ge 1}\) and any histories \(y_{1:t},\,y'_{1:t} \in \mathbb{Z}_{\ge 0}^t\) with \(y_{1:t} \le y'_{1:t}\), we have
\[
\bigl[\,Y_{t+1}\mid Y_{1:t}=y_{1:t}\,\bigr]
\;\le_{\mathrm{LR}}\;
\bigl[\,Y_{t+1}\mid Y_{1:t}=y'_{1:t}\,\bigr].
\]

\end{theorem}

\begin{proof}
By Lemma~\ref{lem:TP2-time-s}, the condition in \eqref{eq:scaled-Lip} implies that the conditional distribution $\P{Y_s = y \mid \theta}$ is TP$_2$ in $(y, \theta)$.
Applying Lemma~\ref{lem.0}, it follows that the joint distribution of
$(Y_1, \ldots, Y_{t+1})$ is MTP$_2$.
Finally, by Lemma~\ref{lem.1}, the MTP$_2$ property implies the desired likelihood ratio ordering:
\[
\bigl[\,Y_{t+1}\mid Y_{1:t}=y_{1:t}\,\bigr]
\;\le_{\mathrm{LR}}\;
\bigl[\,Y_{t+1}\mid Y_{1:t}=y_{1:t}'\,\bigr],
\quad \text{for all } y_{1:t} \le y_{1:t}'.
\]
\end{proof}

\begin{remark}\label{rem.1}
A practical choice of $\lambda:\Real\to\Real$ is the \textbf{softplus} function, defined by
\[
  \lambda(d_t+\theta) = \ln\!\left(1 + e^{d_t+\theta}\right),
\]
which is widely adopted in both statistics and machine learning as a smooth, positive, and numerically stable alternative to the exponential response function; see, for instance, \citet{weiss2022softplus, wiemann2024using}.
One can easily verify that
\[
  \frac{\partial}{\partial \theta} \lambda(d_t+\theta)
  = \frac{ e^{d_t+\theta}}{1 + e^{d_t+\theta}}
  \in (0,1)
\]
for all $d_t\in\Real$ and $\theta \in \Real$.

Regarding the choice of the distribution for the random effect, any specification is admissible since it does not affect the general credibility order established in Theorem~\ref{thm.2}.
As an example, one may consider
\[
\Theta \sim \mathcal{N}(0, \kappa^2)
\]
for some $\kappa^2>0$.
\end{remark}

In \ref{app:ext}, we further consider two extensions of the comonotonic random-effects framework: a Negative Binomial
version of the count component and zero-inflated counterparts.

\section{Empirical study: LGPIF insurance data}\label{sec.6}
To conduct an empirical analysis, we use insurance claim data provided by the Wisconsin Local Government Property Insurance Fund (LGPIF).
The LGPIF offers property insurance to a wide range of public entities, including counties, cities, towns, villages, school districts, and fire departments, as well as organizations classified under a miscellaneous category. The dataset contains claim information and policy characteristics for multiple coverages (collateral) observed from 2006 to 2011. Among these, a subset of variables related to claim information is used in our analysis. In this study, we focus on the LGPIF data restricted to the collision coverage for new and old vehicles.
Collision coverage refers to insurance against losses arising from a vehicle colliding with an object, colliding with another vehicle, or overturning.
For model fitting, we use observations from 2006--2010, and we use the 2011 observations as a validation set to compare the predictive performance of the estimated models.

The 2006--2010 dataset used for model fitting consists of longitudinal observations and includes a total of 1{,}234 local government entities.
We exclude observations for policyholders whose both new-vehicle collision coverage and old-vehicle collision coverage are zeros. As a result, the final longitudinal dataset of $k=497$ entities is used for the analysis.
For a more detailed description of the dataset, see \citet{frees2016multivariate}. As covariates, we use two categorical variables: entity type (six levels: Miscellaneous, City, County, School, Town, and Village) and log coverage (three levels).
The intervals for log coverage are defined as follows: $
\text{coverage 1} \in (0,\;0.17],\
\text{coverage 2} \in (0.17,\;0.75],\
\text{coverage 3} \in (0.75,\;\infty).$ The distributions of entity type and coverage level are summarized in Table~\ref{tab:policy_characteristics_en}. For simplicity, we assume that these covariates do not change over time.

\begin{table}[ht]
\centering
\caption{Characteristics of covariates used in the analysis}
\label{tab:policy_characteristics_en}
\begin{tabular}{llr}
\toprule
\multicolumn{2}{l}{\textbf{Categorical variables}} & \textbf{Proportion (\%)} \\
\midrule
\textbf{Variable} & \textbf{Description} &  \\
\midrule
\multirow{7}{*}{Entity type}
  & Type of local government entity &  \\
  & \quad Miscellaneous & 5.03 \\
  & \quad City          & 9.66 \\
  & \quad County        & 11.47 \\
  & \quad School        & 36.42 \\
  & \quad Town          & 16.90 \\
  & \quad Village       & 20.52 \\
\midrule
\multirow{4}{*}{Coverage level}
  & Collision coverage amount for new and old vehicles &  \\
  & \quad Interval $(0,\;0.14] = 1$        & 33.40 \\
  & \quad Interval $(0.14,\;0.74] = 2$     & 33.20 \\
  & \quad Interval $(0.74,\;\infty) = 3$   & 33.40 \\
\bottomrule
\end{tabular}
\end{table}

\subsection{Model setting}
For comparison, we introduce the following benchmark models (BMs):
\begin{itemize}
  % \item (Model 5) Poisson zero-inflated model with Gaussian random effects: identical to Model~\ref{mod.1} except that we replace the observation equation by $Y_t := Z_t N_t$.

  \item (BM 1) Poisson GLMM:
  \[
  Y_{i,t}\mid R_i,\mathbf{x}_{i} \;\stackrel{\text{iid}}{\sim}\; \mathrm{Pois}\bigr(\lambda_{i}\text{exp}(R_i)\bigr),
  \]
  where $\lambda_{i}=\exp(\langle \mathbf{x}_{i},\boldsymbol\beta \rangle)$, and the random effect prior is assumed to be $R_i\sim N\left(-\frac{1}{2} d^2, d^2\right)$.
  % {\color{red}(Jaeyoun: We failed to have better predictive MSE over this model. It should be good to test whether
  % \[
  % Y_{i,t}\mid R_i,\mathbf{x}_{i} \;\stackrel{\text{iid}}{\sim}\; \mathrm{Pois}(\lambda_{i}\operatorname{exp}{R_i}),
  % \]
  % with $R_i\sim N\left(-\frac{1}{2} d^2, d^2\right)$.
  % )}

  \item (BM 2) Poisson GLM:
  \[
  Y_i\mid \mathbf{x}_i \;\stackrel{\text{iid}}{\sim}\; \mathrm{Pois}(\lambda_i),
  \]
  where $\lambda_i=\exp(\langle \mathbf{x}_{i},\boldsymbol\beta \rangle)$.

  \item (BM 3) Poisson hurdle model:
  \[
  Y_i \mid \mathbf{x}_i \;\stackrel{\text{iid}}{\sim}\; I_i\,Y_i^+,
  \]
  where $I_i \sim \mathrm{Ber}(\pi_i)$ and $Y_i^+ \sim \mathrm{TPois}(\lambda_i)$ with $
  \pi_i = \sigma\!\big(\langle \mathbf{x}_i,\boldsymbol\beta_1\rangle\big),
  \
  \lambda_i = \exp\!\big(\langle \mathbf{x}_i,\boldsymbol\beta_2\rangle\big).
  $

  \item (BM 4) Poisson zero-inflated model:
  \[
  Y_i\mid \mathbf{x}_i \;\stackrel{\text{iid}}{\sim}\; I_i\,Y_i^*,
  \]
  where $I_i \sim \mathrm{Ber}(\pi_i)$ and $Y_i^* \sim \mathrm{Pois}(\lambda_i)$, with $
  \pi_i = \sigma\!\big(\langle \mathbf{x}_{i},\boldsymbol\beta_1 \rangle\big),
  \
  \lambda_i=\exp\!\big(\langle \mathbf{x}_{i},\boldsymbol\beta_2 \rangle\big).
  $
\end{itemize}

Now, we discuss the incorporation of the covariates ($\mathbf{x}_i\in\Real^{p+1}$) into Model \ref{mod.1}, \ref{mod.2}, and \ref{mod.05}.
Although these models are formulated for a single policyholder, we introduce an additional subscript
$i=1,\ldots,k$ to represent the $i$-th policyholder when applying the models to cross-sectional or
panel data.
For Model~\ref{mod.1}, we set
\[
\mu_{1,i} = \langle \mathbf{x}_{i},\boldsymbol\beta_1 \rangle,
\qquad
\mu_{2,i} = \langle \mathbf{x}_{i},\boldsymbol\beta_2 \rangle,
\]
where $\boldsymbol\beta_1,\boldsymbol\beta_2 \in \mathbb{R}^{p+1}$.

For Model~\ref{mod.2}, we incorporate covariates into selected hyperparameters of the random-effect prior by setting
\[
b_i = \exp\!\big(\langle \mathbf{x}_i,\boldsymbol\beta_b\rangle\big),
\qquad
\alpha_i = \exp\!\big(\langle \mathbf{x}_i,\boldsymbol\beta_\alpha\rangle\big),
\]
where $\boldsymbol\beta_b, \boldsymbol\beta_\alpha \in \mathbb{R}^{(p+1)\times 1}$.
Due to convenient conjugacy property, the Poisson–Gamma random–effects model is widely used in insurance applications. In this context, the a priori rate, specific to each policyholder, is employed to capture observable heterogeneity through explanatory covariates (not to be confused with the unobserved heterogeneity modeled via random effects). In the standard Poisson–Gamma formulation, the canonical specification is
\[
N_{i,t}\,\overset{\mathrm{iid}}{\sim}\,\operatorname{Pois}(\lambda_i\,\Theta_{i,2}),
\qquad
\Theta_{i,2}\sim\operatorname{Gamma}(\alpha,\beta).
\]
Since $\lambda_i\Theta^{[2]}_i\sim \operatorname{Gamma}\!\bigl(\alpha,\,\beta/\lambda_i\bigr)$, this specification is equivalent to
\[
N_{i,t}\,\overset{\mathrm{iid}}{\sim}\,\operatorname{Pois}(\Theta_{i,2}^*),
\qquad
\Theta_{i,2}^*\sim\operatorname{Gamma}(\alpha,\beta/\lambda_i).
\]
Thus, the covariates enter through the rate parameter of the Gamma distribution.
By contrast, the covariates in Model~\ref{mod.2} are attached to the shape parameter $\alpha_i$ and to the Beta prior through $b_i$.
This choice greatly simplifies the monotonicity analysis of the predictive mean, as shown in Section \ref{sec.4}.

For Model~\ref{mod.05}, we take $\lambda$ to be the softplus function introduced in Remark~1.
The conditional distributions are specified as
\[
Z_{i,t}\mid \Theta \ \stackrel{\mathrm{iid}}{\sim}\ \mathrm{Ber}\!\big(\sigma(c_i+\Theta_i)\big),
\qquad
N_{i,t}\mid \Theta \ \stackrel{\mathrm{iid}}{\sim}\ \mathrm{Pois}\!\big(\lambda(d_i+\Theta_i)\big),
\]
and the covariates are incorporated through
\[
c_i = \langle \mathbf{x}_i, \boldsymbol\beta_c \rangle,
\qquad
d_i = \langle \mathbf{x}_i, \boldsymbol\beta_d \rangle .
\]
For Model~\ref{mod.1}, Model~\ref{mod.2}, Model~\ref{mod.05}, and BM~1, closed-form likelihood expressions are not available, making classical maximum likelihood estimation difficult.
Posterior inference is therefore carried out using Markov Chain Monte Carlo methods implemented in the \texttt{nimble} package in \texttt{R}, which provides a flexible framework for specifying and fitting hierarchical Bayesian models \citep{de2017programming, nimble}.

%For Model~\ref{mod.1}, Model~\ref{mod.2}, Model~\ref{mod.05}, and BM~1, closed-form likelihood expressions are not available, making classical maximum likelihood estimation difficult.
%We therefore adopt a Bayesian framework for estimation (this sentence can be safely removed. Isn't it?).
%Alternatively, the posterior inference is carried out using Markov chain Monte Carlo methods implemented in the \texttt{nimble} package in \textsf{R} (Do you want to include two sentences to briefly explain nimble?).

\subsection{Numerical investigation of the credibility order}

Before evaluating predictive performance, we first examine whether each fitted model respects the
credibility order.
For some models, the base and/or general credibility orders are guaranteed by construction through
explicit constraints on the model specification.
For BM~1, \citet{purcaru2003dependence} established the
general credibility order. Model~\ref{mod.2} is fitted
under the sufficient condition in Corollary~\ref{thm.3}, which guarantees the base credibility order while the general credibility order is not guaranteed.
For Model~\ref{mod.05}, adopting the softplus specification of $\lambda$ in Remark~\ref{rem.1} satisfies
the condition of Theorem~\ref{thm.2}, and therefore the general credibility order is ensured.
In contrast, BM~2, BM~3, and BM~4 do not include random effects. As a result, their predictive distributions do
not depend on past claim histories beyond the observed covariates. Consequently, the credibility-order comparisons considered
below are trivially satisfied and do not provide meaningful diagnostics for these benchmark models.

For the remaining models, the base and general credibility orders are assessed numerically.
To empirically check the base credibility order, we fix $t=5$ and, for each policyholder $i$, compare
the one-step-ahead posterior predictive means under the two counterfactual last-period outcomes
$Y_{i,5}=0$ and $Y_{i,5}=1$, while holding the earlier observed history $Y_{i,1:4}=y_{i,1:4}$ fixed:
\[
\E{Y_{i,6}\,\middle|\,Y_{i,1:4}=y_{i,1:4},\,Y_{i,5}=0}
\;\le\;
\E{Y_{i,6}\,\middle|\,Y_{i,1:4}=y_{i,1:4},\,Y_{i,5}=1}.
\]
Table~\ref{tab.1} reports the violation rate, defined as the proportion of policyholders for which the
above inequality fails.

%To assess the general credibility order, we conduct analogous comparisons using two standard
%nondecreasing transforms: the deductible (stop-loss) transform $(\cdot-d)^+$ and the limited-loss
%transform $(\cdot\wedge d)$.

%We first examine the base credibility order.
Under Model~\ref{mod.1}, the violation rate is observed to be as high as $9.29\%$.
That is, there are cases in which an individual with a smaller past claim history has a higher expected future claim count than an individual with a larger past claim history.
Such a reversal phenomenon is at an unacceptable level in insurance practice.

\begin{table}[ht]
\centering
\caption{Proportions of monotonicity violations under each credibility order by model}
\label{tab.1}
%\footnotesize
\begin{tabular}{lccccc}
\toprule
 &  & \multicolumn{2}{c}{$(Y-d)^+$} & \multicolumn{2}{c}{$(Y \wedge d)$} \\
 \cmidrule(lr){3-4}\cmidrule(lr){5-6}
 & Base ($Y$) & $d=1$ & $d=2$ & $d=1$ & $d=2$ \\
\midrule
(Model 2) Gaussian Poisson-hurdle   & 0.0929 & 0.1076 & 0.1272 & 0      & 0.0098 \\
(Model 3) Independent Poisson-hurdle & 0      & 0.2421 & 1.0000 & 0      & 0      \\
(Model 4) Comonotonic Poisson-hurdle & 0      & 0      & 0      & 0      & 0      \\
\bottomrule
\end{tabular}
\end{table}

To assess the general credibility order, we conduct analogous comparisons using two standard
nondecreasing loss transformations commonly used in insurance pricing, corresponding to deductible
insurance and policy limit insurance.
Specifically, for each policyholder $i$ (again at $t=5$), we check whether
\[
\E{(Y_{i,6}-d)^+ \,\middle|\, Y_{i,1:4}=y_{i,1:4},\, Y_{i,5}=0}
\;\le\;
\E{(Y_{i,6}-d)^+ \,\middle|\, Y_{i,1:4}=y_{i,1:4},\, Y_{i,5}=1},
\]
and
\[
\E{Y_{i,6}\wedge d \,\middle|\, Y_{i,1:4}=y_{i,1:4},\, Y_{i,5}=0}
\;\le\;
\E{Y_{i,6}\wedge d \,\middle|\, Y_{i,1:4}=y_{i,1:4},\, Y_{i,5}=1}.
\]
Table~\ref{tab.1} summarizes the corresponding violation rates across policyholders and across
different deductible and policy limit levels.

For the posterior expectation of the deductible-transformed claim count, $(Y_6-d)^+$, the reversal rate under Model~\ref{mod.1} increases as the deductible $d$ becomes larger.
However, it is still lower than the corresponding reversal rate under Model~\ref{mod.2}.
Although the base credibility order is guaranteed for Model~\ref{mod.2}, the reversal rate for the posterior expectation of $(Y_6-d)^+$ is already very high at about $24\%$ when $d=1$, and reversals are observed in all cases when $d=2$. Meanwhile, for the posterior expectation of the capped claim amount, $(Y_6 \wedge d)$, reversal phenomena are observed only under Model~\ref{mod.1}.

\subsection{Out-of-sample validation}

According to the out-of-sample validation results for the LGPIF dataset reported in
Table~\ref{tab.2}, models with random effects generally exhibit superior predictive performance.
We define the predictive mean squared error (MSE) as
\[
\mathrm{MSE}
:= \frac{1}{k}\sum_{i=1}^{k}\bigl(y_{i,6}-\hat y_{i,6}\bigr)^2,
% \qquad
% \mathrm{MAE}
% := \frac{1}{k}\sum_{i=1}^{k}\bigl|y_{i,6}-\hat y_{i,6}\bigr|,
\]
where $\hat y_{i,6}$ denotes the predictive mean
\[
\hat y_{i,6}
:= \hat{\mathbb{E}}\!\left[Y_{i,6}\,\mid\,y_{i,1:5}\right]
\]
implied by the fitted model. Here, $\hat{\mathbb{E}}[\cdot\mid y_{i,1:5}]$ denotes the estimated predictive expectation,
computed using the fitted model parameters.

Among the random-effects models, Model~\ref{mod.05} achieves the lowest MSE and, at the same time,
is the only specification that guarantees both the base and the general credibility orders.
Model~\ref{mod.2} (independent random effects) also attains a low MSE and guarantees the base credibility order,
but it does not guarantee the general credibility order.
In contrast, despite its reasonably competitive MSE, Model~\ref{mod.1} fails to guarantee both the base and the general credibility orders,
which undermines its suitability for credibility-based ratemaking.

% \begin{table}[ht]
% \centering
% \caption{Out-of-sample validation results}
% \label{tab.2}
% \begin{tabular}{lccccc}
% \toprule
% \multirow{2}{*}{\quad\quad\quad\quad\quad\quad\textbf{Model}}
%  & \multicolumn{3}{c}{\textbf{Predictive measures}}
%  & \multicolumn{2}{c}{\textbf{Guaranteed credibility order}} \\
% \cmidrule(lr){2-4} \cmidrule(lr){5-6}
%  & \textbf{MSE} & \textbf{MAE} & \textbf{log-likelihood}
%  & \textbf{Base} & \textbf{General} \\
% \midrule
% (Model 2) Gaussian RE Poisson-hurdle       & 1.0489 & 0.5990 & -208.9472 & No & No \\
% (Model 3) Independent RE Poisson-hurdle    & 0.9466 & 0.5951 &-214.1023 & Yes & No \\
% (Model 4) Comonotonic RE Poisson-hurdle    & 0.9382 & 0.5994 &-208.5055 & Yes & Yes \\
% \midrule
% (BM 1) Poisson-Gamma GLMM                          & 1.0561 & 0.6038 &-207.9895 & Yes & Yes \\
% (BM 2) Poisson GLM                                  & 2.0241 & 0.8575 & -266.5104 & - & - \\
% (BM 3) Poisson hurdle                               & 2.4239 & 0.9158 & -301.3224 & - & - \\
% (BM 4) Poisson zero-inflated                        & 1.9762 & 0.8396 & -256.4331 & - & - \\
% \bottomrule
% \end{tabular}
% \end{table}

\begin{table}[ht]
\centering
\caption{Out-of-sample validation results}
\label{tab.2}
\begin{tabular}{lccc}
\toprule
\multirow{2}{*}{\quad\quad\quad\quad\quad\quad\textbf{Model}}
 & \multicolumn{1}{c}{\textbf{Predictive measure}}
 & \multicolumn{2}{c}{\textbf{Guaranteed credibility order}} \\
\cmidrule(lr){2-2} \cmidrule(lr){3-4}
 & \textbf{MSE} & \textbf{Base} & \textbf{General} \\
\midrule
(Model 2) Gaussian RE Poisson-hurdle       & 1.0489 & No & No \\
(Model 3) Independent RE Poisson-hurdle    & 0.9466 & Yes & No \\
(Model 4) Comonotonic RE Poisson-hurdle    & 0.9382 & Yes & Yes \\
\midrule
(BM 1) Poisson GLMM                  & 1.1171 & Yes & Yes \\
(BM 2) Poisson GLM                         & 2.0241 & - & - \\
(BM 3) Poisson hurdle                      & 2.4239 & - & - \\
(BM 4) Poisson zero-inflated               & 1.9762 & - & - \\
\bottomrule
\end{tabular}
\end{table}

In conclusion, considering both predictive accuracy and guaranteed monotonicity-namely, satisfaction of the base and general credibility orders, and even the likelihood-ratio order-the Poisson-hurdle mixture model with comonotonic random effects (Model \ref{mod.05}) can be regarded as the best overall choice.
Therefore, when a hurdle model is adopted as the analysis framework, fitting Model \ref{mod.05} is arguably the most practical option.

% \section{Extensions}\label{sec.7}

\section{Conclusion}\label{sec.8}

This paper examined stochastic monotonicity in counting models for excessive zeros with random effects and demonstrated, both theoretically and numerically, that hurdle specifications with correlated random effects can violate the credibility order. To clarify when this property is preserved, we analyzed a Poisson--hurdle model with comonotonic random effects and showed that it satisfies the credibility order. We then extended this result to the Negative Binomial and zero-inflated counterparts, obtaining parallel monotonicity guarantees.

Several directions remain for future research. First, while extending the framework to dynamic random-effect models is conceptually straightforward, we focused on static random effects for simplicity and clarity of exposition. A natural next step is to develop dynamic specifications in which random effects evolve over time while maintaining stochastic monotonicity. Furthermore, the current formulation is distribution-specific; generalizing the framework to accommodate broader distributional families within the class of excessive-zero models would be another promising direction for future work.

	\section*{Acknowledgements}
	%The authors thank the anonymous referees and editor for their insightful comments, which significantly improved the paper.
	Jae Youn Ahn acknowledges partial support from the Institute of Information \& Communications Technology Planning \& Evaluation (IITP) grant funded by the Korea government (MSIT) (No.RS-2022-00155966, Artificial Intelligence Convergence Innovation Human Resources Development (Ewha Womans University)) and a National Research Foundation of Korea (NRF) grant (2021R1A6A1A10039823, RS-2023-00217022, RS-2025-23524530).

    %\section*{Data Availability Statement}

	%-----------------------------------------------------------------------------------------------------------------------------------------------
	\pagebreak
	\bibliographystyle{apalike}
	\bibliography{bib_tex}
	
	\pagebreak

	\appendix

    \section{Proofs}
\label{app:proofs}

% \section*{Appendix}\label{app.1}

\subsection{Proof of Lemma \ref{lem.2}}
\begin{proof}
We first prove part i. Proof of part i is equivalent with
\begin{equation}\label{eq.1}
f(\theta \mid y_1)f(\theta^* \mid y_1')\le f(\theta\wedge \theta^*\mid y_1)f(\theta\vee \theta^*\mid y_1')
\end{equation}
for any integers \(1 \le y_1 \le y_1'\) and $\theta, \theta^*\in\Real^2$.
Note that \eqref{eq.1} is equivalent with
\begin{equation}\label{eq.2}
f(y_1\mid \theta)\pi(\theta)f(y_1'\mid \theta^*)\pi(\theta^*) \le
f(y_1\mid \theta\wedge \theta^*)\pi(\theta\wedge \theta^*)f(y_1'\mid \theta \vee \theta^*)\pi(\theta \vee \theta^*).
\end{equation}

To show \eqref{eq.2}, we first have
\begin{equation}\label{eq.3}
\pi(\theta)\pi(\theta^*) \le \pi(\theta\wedge \theta^*)\pi(\theta \vee \theta^*).
\end{equation}
$\theta, \theta^*\in\Real^2$ which is from the well known result that \(\Theta\sim \mathcal{N}(\mu,\Sigma)\) is MTP$_2$ if and only if \(\rho\ge0\) \citep{bolviken1982probability, karlin1983m}.

Now we prove
\begin{equation}\label{eq.4}
f(y_1\mid \theta)f(y_1'\mid \theta^*) \le
f(y_1\mid \theta\wedge \theta^*)f(y_1'\mid \theta \vee \theta^*)
\end{equation}
for any integers \(1 \le y_1 \le y_1'\) and $\theta, \theta^*\in\Real^2$.
Moving every term to the right, the equality is equivalent with
It is enough to show
\[
1\le \frac{f(y_1\wedge y_1^*| \theta\wedge \theta^*)f(y_1\vee y_1^*|\theta\vee \theta^*)}{f(y_1|\theta)f(y_1^*|\theta^*)}.
\]
Note that, for integer $y_1\ge 1$,  we have
\[
f(y_1|\theta)=f(z_1=1, n_1=y_1-1\mid \theta)=\sigma(\theta_1) \frac{e^{-{\rm exp}(\theta_2)}e^{\theta_2(y_1-1)}}{(y_1-1)!}.
\]
As a result we have
\[
\begin{aligned}
\frac{f(y_1\wedge y_1^*| \theta\wedge \theta^*)f(y_1\vee y_1^*|\theta\vee \theta^*)}{f(y_1|\theta)f(y_1^*|\theta^*)}
&=\frac{\left(e^{\theta_2}\wedge e^{\theta_2^*}\right)^{y_1\wedge y_1^*-1}}{e^{\theta_2(y_1-1)}}
\frac{\left(e^{\theta_2}\vee e^{\theta_2^*}\right)^{y_1\vee y_1^*-1}}{e^{\theta_2^*(y_1^*-1)}}\\
&=\begin{cases}
  1, & \theta_2\le \theta_2^*, \quad y_1\le y_1^*\\
  \left( \frac{e^{\theta_2^*}}{e^{\theta_2}}\right)^{y_1-y_1^*}\ge 1, & \theta_2\le \theta_2^*, \quad y_1\ge y_1^*\\
  \left( \frac{e^{\theta_2}}{e^{\theta_2^*}}\right)^{y_1^*-y_1}\ge 1, & \theta_2\ge \theta_2^*, \quad y_1\le y_1^*\\
  1, & \theta_2\ge \theta_2^*, \quad y_1\ge y_1^*\\
\end{cases}
\end{aligned}
\]
which concludes the proof of \eqref{eq.4}.

Finally, combining \eqref{eq.3} and \eqref{eq.4}, we conclude the proof of part i.

Now we prove part ii.
By the tower property, we have
\[
\begin{aligned}
\E{Y_2\mid y_1}&=\E{\E{Y_2\mid y_1, \Theta}\mid y_1}\\
&=\E{\E{Y_2\mid \Theta}\mid y_1}
\end{aligned}
\]
where the second equality is from the conditional independence of $Y_1$ and $Y_2$ conditional on $\Theta$.
Hence, by part i and the fact that the likelihood ratio order implies the stochastic order, it is enough to show
\[
g(\theta):=\E{h(Y_2)\mid \Theta=\theta}
\]
is non-decreasing function of $\theta\in\Real^2$ for all non-decreasing function $h:\Real_{\ge 0} \rightarrow \Real$ for which the expectation exists.
By the tower property, we have
\[
\begin{aligned}
g(\theta) &= h(0)\P{Y_2=0\mid \theta} + \E{h(Y_2)\mid \theta, Y_2\ge 1 }\P{Y_2\ge 1\mid \theta}\\
&= h(0)\P{Y_2=0\mid \theta} + \E{h(Y_2)\mid \theta, Z_2= 1 }\P{Y_2\ge 1\mid \theta}\\
&= h(0)\P{Y_2=0\mid \theta} + \E{h(Y_2)\mid \theta }\P{Y_2\ge 1\mid \theta}\\
&=h(0)(1-\sigma(\theta_1)) + \sigma(\theta_1)\E{h(1+N_2)\mid \theta}\\
&=h(0)(1-\sigma(\theta_1)) + \sigma(\theta_1)\E{h(1+N_2)\mid \theta_2}.\\
\end{aligned}
\]
where the third equality is from the conditional independence between $Z_t$ and $N_t$ for given $\Theta$.

We have
\[
\begin{aligned}
\frac{\partial g}{\partial \theta_1}&=\sigma'(\theta_1)\left ( \E{h(1+N_{2}) \mid \theta_2}-h(0)\right)\\
&\ge 0
\end{aligned}
\]
where the inequality follows because $h(1+N_2) \ge h(0)$ almost surely. Furthermore, we also have
\[
\begin{aligned}
\frac{\partial g}{\partial \theta_2}&=\sigma(\theta_1)\frac{\partial}{\partial \theta_2} \E{h(1+N_2) \mid \theta_2}
&\ge 0
\end{aligned}
\]
where the inequality is from the likelihood ratio of Poisson distribution.
Hence, $g:\Real^2\rightarrow \Real$ is a non-decreasing function of $\theta\in\Real^2$, which concludes the proof of part ii.
\end{proof}

\subsection{Proof of Theorem \ref{thm.1}}
\begin{proof}
First, we observe
\begin{equation}\label{eq.p1}
  \begin{aligned}
    \E{Y_2 \mid Y_1=y_1} &= \E{\E{Y_2\mid \Theta} \mid Y_1=y_1}\\
    &=\E{\sigma(\Theta_1)(1+{\rm exp}(\Theta_2)) \mid Y_1=y_1}\\
    &=\frac{\E{\sigma(\Theta_1)(1+{\rm exp}(\Theta_2))1_{[Y_1=y_1]}}}{\P{Y_1=y_1}}\\
    &=\frac{\E{\sigma(\Theta_1)(1+{\rm exp}(\Theta_2))\P{Y_1=y_1 \mid \Theta} }}{\E{\P{Y_1=y_1\mid \Theta}}}\\
  \end{aligned}
\end{equation}
where the first equality is from the tower property and conditional independence between $Y_1$ and $Y_2$ conditional on $\Theta$, and the final equation is again from the tower property.

The characteristic function of the random effects $(\Theta_1, \Theta_2)$ is
\[
\begin{aligned}
\varphi(t)&:=\E{e^{\,i\,(t_1, t_2)\Theta}}\\
&={\rm exp}\left( i(\mu_1, \mu_2) \begin{pmatrix}
                                                                                               t_1 \\
                                                                                               t_2
                                                                                             \end{pmatrix}  +\frac{1}{2} (t_1, t_2) \begin{pmatrix}
                                                                                                                                      \sigma_1^2 & \rho\sigma_1\sigma_2 \\
                                                                                                                                      \rho\sigma_1\sigma_2 & \sigma_2^2
                                                                                                                                    \end{pmatrix}\begin{pmatrix}
                                                                                               t_1 \\
                                                                                               t_2
                                                                                             \end{pmatrix}\right)\\
\end{aligned}
\]
for $t=(t_1,t_2)^\top\in\mathbb{R}^2$.
As $\sigma_1\rightarrow 0^+$, we have %we have $\Sigma_{\epsilon}\to\Sigma_{0}:=\begin{pmatrix}0&0\\0&\sigma_2^2\end{pmatrix}$, hence
\[
\lim\limits_{\sigma_1\rightarrow 0^+}\varphi(t)=
\exp\!\Big(i(\mu_1 t_1+\mu_2 t_2)-\tfrac12\,\sigma_2^2 t_2^2\Big),
\qquad t\in\mathbb{R}^2.
\]
Note that $\exp\!\Big(i(\mu_1 t_1+\mu_2 t_2)-\tfrac12\,\sigma_2^2 t_2^2\Big)$ is the characteristic function of
\((\mu_1,\Theta_2^{*})\) with \(\Theta_2^{*}\sim\mathcal N(\mu_2,\sigma_2^2)\); i.e., the first coordinate is degenerated.
Therefore, by Lévy’s continuity theorem, we have
\begin{equation}\label{eq.p2}
\Theta \stackrel{d}{\rightarrow} \Theta^{*}\qquad\text{as }\sigma_1\to0^+
\end{equation}
where $\stackrel{d}{\rightarrow}$ represent convergence in distribution. Combining \eqref{eq.p1} and \eqref{eq.p2} as well as the dominated convergence theorem, we have
\[
\begin{aligned}
\lim\limits_{\sigma_1\rightarrow 0^+} \E{Y_2 \mid Y_1=0} &= \frac{\E{\sigma(\mu_1) (1+{\rm exp}(\Theta_2^*))(1-\sigma(\mu_1))} }{1-\sigma(\mu_1)}\\
&=\sigma(\mu_1)\left( 1+ {\rm exp}\left( \mu_2 + \frac{1}{2} \sigma_2^2\right)\right)
\end{aligned}
\]
and
\[
\begin{aligned}
\lim\limits_{\sigma_1\rightarrow 0^+} \E{Y_2 \mid Y_1=1}&=
\frac{\E{\sigma(\mu_1) (1+{\rm exp}(\Theta_2^*))\sigma(\mu_1){\rm exp} \left( -e^{\Theta_2^*}\right)} }{\E{\sigma(\mu_1)  {\rm exp}\left( -e^{\Theta_2^*}\right) }}\\
&=
\sigma(\mu_1)\frac{\E{ (1+{\rm exp}(\Theta_2^*)){\rm exp} \left( -e^{\Theta_2^*}\right)} }{\E{  {\rm exp}\left( -e^{\Theta_2^*}\right) }}\\
\end{aligned}
\]
As a result, for $W:=e^{\Theta_2^*}$, we have
\[
\begin{aligned}
\lim\limits_{\sigma_1\rightarrow 0^+} \left(\E{Y_2 \mid Y_1=0}-\E{Y_2 \mid Y_1=1}\right)
&=\lim\limits_{\sigma_1\rightarrow 0^+} \E{Y_2 \mid Y_1=0}-\lim\limits_{\sigma_1\rightarrow 0^+} \E{Y_2 \mid Y_1=1}\\
&=\sigma(\mu_1)\left( \E{W} - \frac{\E{We^{-W}}}{\E{e^{-W}}}\right)\\
&= - \sigma(\mu_1)\left( \frac{\E{We^{-W}}- \E{W}\E{e^{-W}}}{\E{e^{-W}}}\right)\\
&= - \sigma(\mu_1)\left( \frac{\cov{W, e^{-W}}}{\E{e^{-W}}}\right)\\
&>0
\end{aligned}
\]
where the inequality is from the fact that $e^{-W}$ is a strictly decreasing function of $W$ and the fact that $W$ is nondegenerate, hence $\operatorname{Cov}(X,e^{-X})<0$.

\end{proof}

% \section{Auxiliary Results}
% \label{app:Lemma}

\section{Extensions}
\label{app:ext}
In Section~\ref{sec.5}, we examined the Poisson--hurdle model with comonotonic random effects, which was shown to satisfy the general credibility order.
This section extends that framework in two directions:
first, by replacing the Poisson component with a Negative Binomial specification, and second, by introducing zero-inflated counterparts.
For each extension, we derive the corresponding conditions under which the general credibility order property is preserved.

\subsection{Extension to the Negative Binomial model}

The following specification modifies Model~\ref{mod.05} by replacing the Poisson positive-count component with a Negative Binomial component, and the subsequent result establishes that this model satisfies the general credibility order.

\begin{model}\label{mod.6}
For the given exogenous
\[
(c_t)_{t\ge 1}\subset \Real \quad\hbox{and}\quad (d_t)_{t\ge 1}\subset \Real,
\]
the \textbf{Negative Binomial–hurdle mixture model with comonotonic random effects}for $(Y_{t})_{t\ge 1}$ is defined as
		\[
		Y_{t} := Z_{t}(1+N_{t}), \quad t \ge 1,
		\]
		under the following assumptions:
		
		\begin{itemize}
\item[i.] \textbf{Observation distribution.}
Conditional on $\Theta$, the pairs $(Z_t,N_t)$ are independent \ across $t$, and
$Z_t \perp N_t \mid \Theta$.
Their conditional distributions are
\begin{equation}\label{eq.d.10}
Z_t \mid \Theta \;\sim\; \mathrm{Ber}\!\big(\sigma(c_t + \Theta)\big),
\qquad
N_t \mid \Theta \;\sim\; \mathrm{NB}\!\Big(r,\, \sigma\!\big(d_t + \Theta/r\big)\Big).
\end{equation}

\item[ii.] \textbf{Random effect distribution.}
Assume
\[
\Theta\sim \operatorname{N}(0, s^2).
\]
for a constant $s^2>0$.
		\end{itemize}
		
	\end{model}

\begin{theorem}\label{thm.22}
Consider the setting of  Model~\ref{mod.6}.
Then, for any \(t \in \mathbb{Z}_{\ge 1}\) and any histories \(y_{1:t},\,y'_{1:t} \in \mathbb{Z}_{\ge 0}^t\) with \(y_{1:t} \le y'_{1:t}\), we have
\[
\bigl[\,Y_{t+1}\mid Y_{1:t}=y_{1:t}\bigr]
\;\le_{\mathrm{LR}}\;
\bigl[\,Y_{t+1}\mid Y_{1:t}=y'_{1:t}\bigr].
\]
\end{theorem}

\begin{proof}
  Under the similar logic as in the proof of Theorem~\ref{thm.2}, it is enough to show that
\[
\gamma(\theta, y_j)
:=\frac{\P{Y_j=y_j+1 \mid \Theta=\theta}}{\P{Y_j=y_j \mid \Theta=\theta}}
\]
is the non-decreasing function of $\theta$ for fixed $y_j\in\mathbb{Z}_{\ge 0}$.

We first investigate the case $y_j=0$.
By Model~\ref{mod.6}, we have
\[
\P{Y_j=0\mid \Theta=\theta}=1-\sigma(c_j+\theta),
\]
and
\[
\P{Y_j=1\mid \Theta=\theta}
=\sigma(c_j+\theta)\,\bigl(1-\sigma(d_j+\theta/r)\bigr)^r.
\]
Hence,
\[
\gamma(\theta,0)
=\frac{\sigma(c_j+\theta)}{1-\sigma(c_j+\theta)}\,\bigl(1-\sigma(d_j+\theta/r)\bigr)^r
=\exp(c_j+\theta)\,\bigl(1-\sigma(d_j+\theta/r)\bigr)^r.
\]
Since $\gamma(\theta,0)>0$, $\frac{\partial }{\partial \theta} \gamma(\theta,0)\ge 0$ is equivalent to
\[
\frac{\partial}{\partial \theta}\log \gamma(\theta,0)\ge 0.
\]
Since we have
\[
\begin{aligned}
\frac{\partial}{\partial \theta}\log \gamma(\theta,0)
&= \frac{\partial}{\partial \theta}\Bigl\{(c_j+\theta)+r\log\bigl(1-\sigma(d_j+\theta/r)\bigr)\Bigr\}\\
&= 1 + r\,\frac{\partial}{\partial \theta}\log\bigl(1-\sigma(d_j+\theta/r)\bigr)\\
&= 1 - \sigma(d_j+\theta/r),
\end{aligned}
\]
where the last equality follows from
$\frac{\partial}{\partial x}\log(1-\sigma(x))=-\sigma(x)$.
Therefore,
\[
\frac{\partial}{\partial \theta}\log \gamma(\theta,0)\ge 0, \qquad \hbox{for all } \theta\in\Real,
\]
which implies that $\gamma(\theta,0)$ is non-decreasing in $\theta$.

In the case of $y_j\in\mathbb{Z}_{\ge 1}$, by the recursion of the Negative Binomial distribution, we have
\[
\gamma(\theta, y_j)
=\frac{y_j+r-1}{y_j}\,\sigma(d_j+\theta/r),
\]
which is an increasing function of $\theta\in\Real$.

Hence $\gamma(\theta,y_j)$ is non-decreasing in $\theta$ for all $y_j\in\mathbb{Z}_{\ge 0}$, which concludes the proof.
\end{proof}

\subsection{Extension to zero-inflated models}

The following specification modifies Model~\ref{mod.05} by replacing the hurdle structure with a zero-inflated formulation, in which structural zeros are mixed with the standard count component.
The subsequent result establishes that this model satisfies the general credibility order.

\begin{model}\label{mod.7}
For the given exogenous
\[
(c_t)_{t\ge 1}\subset \Real \quad\hbox{and}\quad (d_t)_{t\ge 1}\subset \Real,
\]
the \textbf{Poisson–zero-inflated mixture model with comonotonic random effects} for $(Y_{t})_{t\ge 1}$ is defined as
		\[
		Y_{t} := Z_{t}N_{t}, \quad t \ge 1,
		\]
		under the following assumptions:
		
		\begin{itemize}
\item[i.] \textbf{Observation distribution.}
Conditional on $\Theta$, the pairs $(Z_t,N_t)$ are independent \ across $t$, and
$Z_t \perp N_t \mid \Theta$.
Their conditional distributions are
\begin{equation}\label{eq.d.11}
Z_t \mid \Theta \;\sim\; \mathrm{Ber}\!\big(\sigma\left(c_t + \Theta\right)\big),
\qquad
N_t \mid \Theta \;\sim\; \mathrm{Poisson}\!\big(\lambda(d_t+\Theta)\big),
\end{equation}
where $\lambda: \Real \to \Real_{>0}$ is a non-decreasing differentiable function.

\item[ii.] \textbf{Random effect distribution.}
Let $\Theta$ be a real-valued random variable with distribution $\mathcal{L}(\Theta)$
supported on $\mathcal{R}(\Theta)\subseteq\mathbb{R}$.

		\end{itemize}
		
	\end{model}

\begin{theorem}\label{thm.201}
Consider the setting of  Model~\ref{mod.7}.
Suppose that, for all $\theta$ on the support of $\Theta$, we have the assumption in \eqref{eq:scaled-Lip}.
Then, for any \(t \in \mathbb{Z}_{\ge 1}\) and any histories \(y_{1:t},\,y'_{1:t} \in \mathbb{Z}_{\ge 0}^t\) with \(y_{1:t} \le y'_{1:t}\), we have
\[
\bigl[\,Y_{t+1}\mid Y_{1:t}=y_{1:t}\bigr]
\;\le_{\mathrm{LR}}\;
\bigl[\,Y_{t+1}\mid Y_{1:t}=y_{1:t}'\bigr].
\]
%
%Then, for every $t\ge1$, we have:
%\begin{itemize}
%\item[i.]
%For any $y_{1:t},y_{1:t}'\in \mathbb{Z}_{\ge 0}^t$ with $y_{1:t}\le y_{1:t}'$,
%\[
%\bigl[\,\Theta\mid Y_{1:t}=y_{1:t}\,\bigr]
%\;\le_{\mathrm{LR}}\;
%\bigl[\,\Theta\mid Y_{1:t}=y_{1:t}'\,\bigr].
%\]
%
%
%\item[ii.]
%For any $y_{1:t},y_{1:t}'\in\mathbb{Z}_{\ge 0}^t$ with $y_{1:t}\le y_{1:t}'$,
%\[
%\bigl[\,Y_{t+1}\mid Y_{1:t}=y_{1:t}\bigr]
%\;\le_{\mathrm{LR}}\;
%\bigl[\,Y_{t+1}\mid Y_{1:t}=y_{1:t}'\bigr].
%\]
%
%\end{itemize}
\end{theorem}
\begin{proof}
    Under the similar logic as in the proof of Theorem \ref{thm.2}, it is enough to show that
\[
\begin{aligned}
\gamma(\theta, y)
%&:=\frac{\pi(\theta \mid y_{-j},\,y)}{\pi(\theta \mid y_{-j},\,y+1)}\\
&:=\frac{\P{Y_j=y+1 \mid \theta}}{\P{Y_j=y \mid \theta}}
\end{aligned}
\]
is the non-decreasing function of $\theta\in \mathcal{R}(\Theta)$ for fixed $y\in\mathbb{Z}_{\ge 0}$.

For $y\in\mathbb{Z}_{\ge 1}$, we have
\[
\gamma(\theta, y)=\frac{\lambda(d_j+\theta)}{y+1}
\]
which is clearly the non-decreasing function of $\theta\in \mathcal{R}(\Theta)$.

It remains to show the case when $y=0$. When $y=0$, we have
\begin{equation}\label{eq.p12}
1/\gamma(\theta, y)=\frac{\P{Z_j=0 \mid \theta}}{\P{Z_j=1 \mid \theta}\P{N_j=1 \mid \theta}} + \frac{\P{N_j=0 \mid \theta}}{\P{N_j=1 \mid \theta}}
\end{equation}
Since the second term in \eqref{eq.p12} is non-increasing function of $\theta$, it is enough to show the first term in \eqref{eq.p12} is non-increasing function of $\theta$.
The first term in \eqref{eq.p12} can be written as
\begin{equation}\label{eq.p13}
\frac{\P{Z_j=0 \mid \theta}}{\P{Z_j=1 \mid \theta}\P{N_j=1 \mid \theta}}=
\left(\frac{\P{Z_j=0 \mid \theta}}{\P{Z_j=1 \mid \theta}\P{N_j=0 \mid \theta}}\right) \frac{1}{\lambda(d_t+\theta)}.
\end{equation}
By the assumption in \eqref{eq:scaled-Lip} and Lemma \ref{lem:TP2-time-s}, the first multiplicative term in \eqref{eq.p13}
\[
\frac{\P{Z_j=0 \mid \theta}}{\P{Z_j=1 \mid \theta}\P{N_j=0 \mid \theta}}
\]
a non-increasing function of $\theta$. Since the second multiplicative term in \eqref{eq.p13} is also non-increasing function of $\theta$, we conclude that \eqref{eq.p13} is also a non-increasing function of $\theta$, which in turn implies \eqref{eq.p12} is also non-increasing function of $\theta$.

\end{proof}

\section{Additional Examples}
\label{app:Example}

\subsection{Counterexample to Lemma~\ref{lem.1}}
Violations of the base credibility order also occur when a negative correlation exists between the random effects.
Table~\ref{rho_simulation} reports Monte Carlo estimates ($S=300{,}000$) obtained by varying $\rho$ while fixing $\sigma_1=\sigma_2=0.5$, $\mu_1=0$, and $\mu_2=-2$.
When $\rho$ becomes negative, reversals in the posterior expectation are observed, which provides a counterexample to Lemma~\ref{lem.1}.

\begin{table}[ht]
\centering
\caption{Comparison of
$\E{Y_2 \mid Y_1=0}$ and $\E{Y_2 \mid Y_1=1}$ as $\rho$ decreases}
\label{rho_simulation}
\begin{tabular}{ccc}
\toprule
$\rho$ & $\mathbb{E}[Y_2 \mid Y_1 = 1]$ & $\mathbb{E}[Y_2 \mid Y_1 = 2]$ \\
\midrule
$-0.8$ & 0.6277 (0.0013) & 0.5711 (0.0011) \\
$-0.5$ & 0.6311 (0.0014) & 0.6104 (0.0014) \\
$0$    & 0.6399 (0.0015) & 0.6874 (0.0021) \\
$0.5$  & 0.6447 (0.0019) & 0.7566 (0.0020) \\
$0.8$  & 0.6493 (0.0023) & 0.7930 (0.0025) \\
\bottomrule
\end{tabular}
\caption*{\footnotesize Note: Parentheses report MCSEs.}
\end{table}

\subsection{Examples of Remark~\ref{rmk.1}}
%In the zero-inflated model with correlated random effects described in 
As briefly mentioned in Remark~\ref{rmk.1}, a violation of the credibility order also arises depending on the magnitude of the random-effects variances in the zero-inflated model, Model \ref{mod.7}. The examples below present simulation results for a zero-inflated model in which the random effects \(\Theta_1\) and \(\Theta_2\) are assumed to follow a bivariate normal distribution, as in Model~\ref{mod.1}. Table~\ref{zip_sigma1_simulation} reports Monte Carlo estimates of
$\mathbb{E}[Y_{2}\mid Y_{1}=0]$ and $\mathbb{E}[Y_{2}\mid Y_{1}=1]$
obtained by varying $\sigma_1^2$ while fixing $\mu_1=\mu_2=0$, $\sigma_2^2=1$, and $\rho=0.5$.
As $\sigma_1^2$ decreases, the ordering between the two conditional expectations is reversed, namely, $\E{Y_2 \mid Y_1=1}\le \E{Y_2 \mid Y_1=0}$. Table~\ref{zip_sigma_2_simulation} presents Monte Carlo estimates obtained by varying $\sigma_{2}^2$
while fixing $\mu_1=\mu_2=0$, $\sigma_{1}^2=1$, and $\rho=0.5$.
As $\sigma_{2}^2$ increases, violations of the base credibility order are again observed.

\begin{table}[ht]
\centering
\caption{Comparison of
$\E{Y_2 \mid Y_1=0}$ and $\E{Y_2 \mid Y_1=1}$ as $\sigma_{1}^2$ decreases}
\label{zip_sigma1_simulation}
\begin{tabular}{ccc}
\toprule
$\sigma_{1}^2$ & $\mathbb{E}[Y_2 \mid Y_1 = 0]$ & $\mathbb{E}[Y_2 \mid Y_1 = 1]$ \\
\midrule
$2.00$ & 0.5426 (0.0077) & 0.8240 (0.0067) \\
$1.00$ & 0.5988 (0.0081) & 0.7503 (0.0056) \\
$0.10$ & 0.6846 (0.0094) & 0.6032 (0.0041) \\
$0.01$ & 0.6849 (0.0097) & 0.5675 (0.0033) \\
\bottomrule
\end{tabular}
\caption*{\footnotesize Note: Parentheses report MCSEs.}
\end{table}

\begin{table}[ht]
\centering
\caption{Comparison of
$\E{Y_2 \mid Y_1=0}$ and $\E{Y_2 \mid Y_1=1}$ as $\sigma_{2}^2$ increases}
\label{zip_sigma_2_simulation}
\begin{tabular}{ccc}
\toprule
$\sigma_{2}^2$ & $\mathbb{E}[Y_2 \mid Y_1 = 0]$ & $\mathbb{E}[Y_2 \mid Y_1 = 1]$ \\
\midrule
0.1 & 0.4742 (0.0027) & 0.6740 (0.0030) \\
1.0 & 0.5988 (0.0081) & 0.7503 (0.0056) \\
2.0 & 0.8857 (0.0282) & 0.7278 (0.0053) \\
3.0 & 1.2649 (0.0489) & 0.6958 (0.0052) \\
\bottomrule
\end{tabular}
\caption*{\footnotesize Note: Parentheses report MCSEs.}
\end{table}

\end{document}